%% file: main.tex
\RequirePackage{luatex85}

\documentclass[11pt,a4paper]{article}
\usepackage{iftex}

\ifPDFTeX
\usepackage[utf8]{inputenc}
\usepackage[T1]{fontenc}
\fi

\usepackage{amsmath}
\usepackage{amssymb}
\usepackage{amsthm}
\usepackage{thmtools}

\ifPDFTeX
\usepackage{libertine}
\usepackage[libertine]{newtxmath} 
\usepackage[scaled=0.96]{zi4} 
\else
\usepackage{fontspec}
\usepackage{unicode-math}
\fi

\usepackage[DIV=11]{typearea} 

\usepackage{microtype}

\usepackage[procnumbered,ruled,vlined,linesnumbered]{algorithm2e}

\usepackage{tikz}

\usepackage{xcolor}
\usepackage{xspace}
\usepackage{enumitem}

\usepackage{comment}

\usepackage[
	style=alphabetic,
	backref=true,
	doi=false,
	url=false,
	maxcitenames=3,
	mincitenames=3,
	maxbibnames=10,
	minbibnames=10,
	backend=biber,
	sortlocale=en_US
]{biblatex}

\usepackage[ocgcolorlinks]{hyperref} 
\usepackage{cleveref}

\input{commands}

\colorlet{DarkRed}{red!50!black}
\colorlet{DarkGreen}{green!50!black}
\colorlet{DarkBlue}{blue!50!black}

\hypersetup{
	linkcolor = DarkRed,
	citecolor = DarkGreen,
	urlcolor = DarkBlue,
	bookmarks = true,
	bookmarksnumbered = true,
	linktocpage = true
}

\allowdisplaybreaks[1]

\declaretheorem[numberwithin=section]{theorem}
\declaretheorem[numberlike=theorem]{lemma}
\declaretheorem[numberlike=theorem]{proposition}

\declaretheorem[numberlike=theorem]{claim}
\declaretheorem[numberlike=theorem]{observation}

\DontPrintSemicolon
\SetKw{KwAnd}{and}
\SetProcNameSty{textsc}
\SetFuncSty{textsc}

\usetikzlibrary{decorations.pathreplacing}

\newcommand{\dist}{\ensuremath{\operatorname{dist}}} 
\newcommand{\distest}{\delta}
\newcommand{\lev}{\ell}
\newcommand{\ball}{\mathit{Ball}}
\newcommand{\p}{k}

\newcommand{\polylog}{\operatorname{polylog}}

\title{Decremental Single-Source Shortest Paths on Undirected Graphs in Near-Linear Total Update Time\thanks{Accepted to \emph{Journal of the ACM}. A preliminary version of this paper was presented at the \emph{55th IEEE Symposium on Foundations of Computer Science (FOCS 2014)}.}}

\author{
Monika Henzinger\thanks{University of Vienna, Faculty of Computer Science, Austria. Supported by the Austrian Science Fund (FWF): P23499-N23. The research leading to these results has received funding from the European Research Council under the European Union's Seventh Framework Programme (FP/2007-2013) / ERC Grant Agreement no.~340506 and from the European Union's Seventh Framework Programme (FP7/2007-2013) under grant agreement no.~317532.}
\and Sebastian Krinninger\thanks{University of Salzburg, Department of Computer Sciences, Austria. Supported by the University of Vienna (IK \mbox{I049-N}). Work done in large part while at University of Vienna, Austria.}
\and Danupon Nanongkai\thanks{KTH Royal Institute of Technology,  School of Electrical Engineering and Computer Science (EECS), Sweden. Work partially done while at ICERM, Brown University, USA, and Nanyang Technological University, Singapore 637371, and while supported in part by the following research grants: Nanyang Technological University grant M58110000, Singapore Ministry of Education (MOE) Academic Research Fund (AcRF) Tier 2 grant MOE2010-T2-2-082, and Singapore MOE AcRF Tier 1 grant MOE2012-T1-001-094.}
}

\date{}

\hypersetup{
	pdftitle = {Decremental Single-Source Shortest Paths on Undirected Graphs in Near-Linear Total Update Time},
	pdfauthor = {Monika Henzinger, Sebastian Krinninger, Danupon Nanongkai}
}

\begin{document}
\maketitle
\begin{abstract}
\input{abstract}
\end{abstract}
\newpage

\tableofcontents
\newpage

\input{intro}
\input{prelim}
\input{overview}
\input{sssp_to_balls}
\input{balls_to_sssp}
\input{putting_together}
\input{conclusion}

\section*{Acknowledgement}
The authors would like to thank the anonymous reviewers of FOCS and JACM for their valuable feedback.

\printbibliography[heading=bibintoc] 

\end{document}

%% file: commands.tex
\makeatletter
\g@addto@macro\bfseries{\boldmath}
\makeatother

\makeatletter
\g@addto@macro\mdseries{\unboldmath}
\g@addto@macro\normalfont{\unboldmath}
\g@addto@macro\rmfamily{\unboldmath}
\g@addto@macro\upshape{\unboldmath}
\makeatother

\DeclareCiteCommand{\citem}
    {}
    {\mkbibbrackets{\bibhyperref{\usebibmacro{postnote}}}}
    {\multicitedelim}
    {}

\renewcommand*{\multicitedelim}{\addcomma\space}

\AtEveryCitekey{\clearfield{note}}

\newcommand{\myhref}[1]{%
  \iffieldundef{doi}
    {\iffieldundef{url}
       {#1}
       {\href{\strfield{url}}{#1}}}
    {\href{http://dx.doi.org/\strfield{doi}}{#1}}%
}

\DeclareFieldFormat{title}{\myhref{\mkbibemph{#1}}}
\DeclareFieldFormat
  [article,inbook,incollection,inproceedings,patent,thesis,unpublished]
  {title}{\myhref{\mkbibquote{#1\isdot}}}

\makeatletter
\AtBeginDocument{%
    \newlength{\temp@x}%
    \newlength{\temp@y}%
    \newlength{\temp@w}%
    \newlength{\temp@h}%
    \def\my@coords#1#2#3#4{%
      \setlength{\temp@x}{#1}%
      \setlength{\temp@y}{#2}%
      \setlength{\temp@w}{#3}%
      \setlength{\temp@h}{#4}%
      \adjustlengths{}%
      \my@pdfliteral{\strip@pt\temp@x\space\strip@pt\temp@y\space\strip@pt\temp@w\space\strip@pt\temp@h\space re}}%
    \ifpdf
      \typeout{In PDF mode}%
      \def\my@pdfliteral#1{\pdfliteral page{#1}}
      \def\adjustlengths{}%
    \fi
    \ifxetex
      \def\my@pdfliteral #1{}
      \def\adjustlengths{\setlength{\temp@h}{-\temp@h}\addtolength{\temp@y}{1in}\addtolength{\temp@x}{-1in}}%
    \fi%
    \def\Hy@colorlink#1{%
      \begingroup
        \ifHy@ocgcolorlinks
          \def\Hy@ocgcolor{#1}%
          \my@pdfliteral{q}%
          \my@pdfliteral{7 Tr}
        \else
          \HyColor@UseColor#1%
        \fi
    }%
    \def\Hy@endcolorlink{%
      \ifHy@ocgcolorlinks%
        \my@pdfliteral{/OC/OCPrint BDC}%
        \my@coords{0pt}{0pt}{\pdfpagewidth}{\pdfpageheight}%
        \my@pdfliteral{F}
        \my@pdfliteral{EMC/OC/OCView BDC}%
        \begingroup%
          \expandafter\HyColor@UseColor\Hy@ocgcolor%
          \my@coords{0pt}{0pt}{\pdfpagewidth}{\pdfpageheight}%
          \my@pdfliteral{F}
        \endgroup%
        \my@pdfliteral{EMC}%
        \my@pdfliteral{0 Tr}
        \my@pdfliteral{Q}%
      \fi
      \endgroup
    }%
}
\makeatother

%% file: abstract.tex
In the decremental single-source shortest paths (SSSP) problem we want to maintain the distances between a given source node $s$ and every other node in an $n$-node $m$-edge graph~$G$ undergoing edge deletions.
While its static counterpart can be solved in near-linear time, this decremental problem is much more challenging even in the {\em undirected unweighted} case.
In this case, the classic $O(mn)$ total update time of Even and Shiloach~[JACM 1981] has been the fastest known algorithm for three decades. At the cost of a $(1+\epsilon)$-approximation factor, the running time was recently improved to $ n^{2+o(1)} $ by Bernstein and Roditty~[SODA 2011].
In this paper, we bring the running time down to near-linear: 
We give a $(1+\epsilon)$-approximation algorithm with $ m^{1+o(1)} $ expected total update time, thus obtaining {\em near-linear time}.
Moreover, we obtain $ m^{1+o(1)} \log W $ time for the weighted case, where the edge weights are integers from $ 1 $ to $ W $.
The only prior work on weighted graphs in $ o(m n) $ time is the $ m n^{0.9 + o(1)} $-time algorithm by Henzinger et al.~[STOC 2014, ICALP 2015] which works for directed graphs with quasi-polynomial edge weights.
The expected running time bound of our algorithm holds against an oblivious adversary.

In contrast to the previous results which rely on maintaining a sparse emulator, our algorithm relies on maintaining a so-called {\em sparse $(h, \epsilon)$-hop set} introduced by Cohen~[JACM 2000] in the PRAM literature.
An $(h, \epsilon)$-hop set of a graph $G=(V, E)$ is a set $F$ of weighted edges such that the distance between any pair of nodes in $G$ can be $(1+\epsilon)$-approximated by their $h$-hop distance (given by a path containing at most $h$ edges) on $G'=(V, E\cup F)$.
Our algorithm can maintain an $(n^{o(1)}, \epsilon)$-hop set of near-linear size in near-linear time under edge deletions.
It is the first of its kind to the best of our knowledge.
To maintain approximate distances using this hop set, we extend the monotone Even-Shiloach tree of Henzinger et al.~[SICOMP 2016] and combine it with the bounded-hop SSSP technique of Bernstein~[FOCS 2009, STOC 2016] and M\k{a}dry~[STOC 2010].
These two new tools might be of independent interest.

%% file: intro.tex
\section{Introduction}\label{sec:intro}

Dynamic graph algorithms refer to data structures on graphs that support update and query operations. They are classified according to the type of update operations they allow: {\em decremental} algorithms allow only edge deletions, {\em incremental} algorithms allow only edge insertions, and {\em fully dynamic} algorithms allow both insertions and deletions. 
In this paper, we consider decremental algorithms for the {\em single-source shortest paths (SSSP)} problem on {\em undirected} graphs. The {\em unweighted} case of this problem allows the following operations. 
\begin{itemize}
\setlength{\itemsep}{0.5pt}
\item {\sc Delete}($u, v$): delete the edge $(u, v)$ from the graph, and
\item {\sc Distance}($v$): return the distance $ \dist_G (s, v) $ between node $s$ and node $v$ in the current graph~$G$.
\end{itemize} 
The {\em weighted} case allows an additional operation {\sc Increase($ u, v, \Delta $)} which increases the weight of the edge $ (u, v) $ by $ \Delta $.
We allow positive integer edge weights in the range from $ 1 $ to $ W $, for some parameter~$ W $.
For any $\alpha\geq 1$, we say that an algorithm is an {\em $\alpha$-approximation} algorithm if, for any distance query {\sc Distance}($x$), it returns a distance estimate $\distest (s, x)$ such that $\dist_{G}(s, x)\leq \distest (s, x) \leq \alpha\dist_G(s, x)$.
There are two time complexity measures associated with this problem: {\em query time} denoting the time needed to answer {\em each} distance query, and {\em total update time} denoting the time needed to process {\em all} edge deletions. The running time will be measured in terms of $n$, the number of nodes in the graph, and $m$, the number of edges {\em before} the first deletion. For the weighted case, we additionally consider the dependence on $W$, the maximum edge weight.
We use $\tilde O$-notation to hide factors that are polylogarithmic in~$ n $.
In this paper, we focus on algorithms with small ($O(1)$ or $ \polylog n $) query time, and the main goal is to minimize the total update time, which will simply be referred to as {\em time} when the context is clear.

\paragraph{Related Work}
The static version of SSSP can be easily solved in $ \tilde O (m)$ time using, e.g., Dijkstra's algorithm. Moreover, due to the deep result of Thorup~\cite{ThorupJACM99}, it can even be solved in linear ($O(m)$) time in undirected graphs with positive integer edge weights.
This implies that in our setting we can naively solve decremental SSSP in $O(m^2)$ total update time by running the static algorithm after every deletion. The first non-trivial decremental algorithm is due to Even and Shiloach~\cite{EvenS81} from 1981 and takes $O(mn)$ total update time in unweighted undirected graphs.  
This algorithm will be referred to as {\em ES-tree} throughout this paper. It has many applications such as for decremental strongly connected components~\cite{Roditty13} and multicommodity flow problems~\cite{Madry10}; yet, the ES-tree has resisted many attempts of improving it for decades.
Roditty and Zwick~\cite{RodittyZ11} explained this phenomenon by providing evidence that the ES-tree is optimal for maintaining exact distances even on {\em unweighted undirected} graphs, unless there is a major breakthrough for Boolean matrix multiplication and many other long-standing problems \cite{WilliamsW10}. 
After the preliminary version of our work appeared, Henzinger et al.~\cite{HenzingerKNS15} showed that, up to subpolynomial factors, $ O(mn) $ is essentially the best possible total update time for maintaining exact distances under the assumption that there is no ``truly subcubic'' algorithm for a problem called online Boolean matrix-vector multiplication.
Under the same assumption, they also showed that there is no fully dynamic $\alpha$-approximate SSSP algorithm such that $ \alpha < 2 $ with amortized time $ O (m^{\gamma - \delta}) $ \emph{per update} and query time $ O (m^{1 - \gamma - \delta}) $ for any $ \gamma \in (0, 1) $ and $ \delta > 0 $.\footnote{This conditional lower bound then holds for graphs with $ m \leq \min(n^{1/\gamma}, n^{1/(1 - \gamma)}) $ many edges.}
In incremental and decremental algorithms, respectively, the same type of trade-off holds between the \emph{worst-case} update time and the query time.
It is thus natural to shift the focus to {\em amortized decremental approximation algorithms}; the amortization is usually done implicitly by only considering the \emph{total} update time over a sequence of up to $ m $ deletions.

The first improvement for unweighted undirected graphs was due to Bernstein and Roditty~\cite{BernsteinR11} who presented a randomized $(1+\epsilon)$-approximation algorithm with $ n^{2+O(1/\sqrt{\log n})} $ total update time.\footnote{To enhance readability we assume that $ \epsilon $ is a constant when citing related work, thus omitting the dependence on $ \epsilon $ in the running times.}
This time bound is only slightly larger than quadratic and beats the $O(mn)$ time of the ES-tree unless the input graph is very sparse. 
After the preliminary version of our work appeared, Bernstein and Chechik, presented deterministic $(1+\epsilon)$-approximation algorithms for unweighted undirected graphs with total update times $ \tilde O (n^2) $~\cite{BernsteinC16} and $ \tilde O (n^{1.25} \sqrt{m}) = \tilde O (m n^{3/4}) $~\cite{BernsteinC17}, respectively.
In weighted undirected graphs, an extension of the technique gives a total update time of $ \tilde O (n^2 \log{W}) $~\cite{Bernstein17}.

For the case of directed graphs, Henzinger and King~\cite{HenzingerK95} observed that the ES-tree can be easily adapted to unweighted directed graphs. King~\cite{King99} later extended the ES-tree to an $O(mnW)$-time algorithm for weighted directed graphs. A rounding technique used in recent algorithms of Bernstein~\cite{Bernstein09,Bernstein16} and M\k{a}dry~\cite{Madry10}, as well as earlier papers on approximate shortest paths~\cite{KleinS97,Cohen98,Zwick02}, gives a $(1+\epsilon)$-approximate $\tilde O(mn\log W)$-time algorithm for weighted directed graphs.
Very recently, we obtained a randomized $(1+\epsilon)$-approximation algorithm with total update time $ m n^{0.9 + o(1)} $ for decremental approximate SSSP in weighted directed graphs~\cite{HenzingerKNSTOC14,HenzingerKNICALP15} if $ W \leq 2^{\log^c{n}}$ for some constant~$ c $.
This gives the first $ o(m n)$-time algorithm for the directed case, as well as other important problems such as single-source reachability and strongly connected components \cite{RodittyZ08,Lacki13,Roditty13,ChechikHILP16}.
Also very recently, Abboud and Williams~\cite{AbboudW14} showed that ``deamortizing'' our algorithms in \cite{HenzingerKNSTOC14} might not be possible: a combinatorial algorithm with {\em worst case} update time and query time of $ O(n^{2-\delta})$ (for any $ \delta > 0 $) per deletion implies a faster combinatorial algorithm for Boolean matrix multiplication and, for the more general problem of maintaining the number of reachable nodes from a source under deletions (which our algorithms in \cite{HenzingerKNSTOC14} can do) a worst case update and query time of $ O(m^{1-\delta}) $ (for any $ \delta > 0 $) will falsify the strong exponential time hypothesis.

\paragraph{Our Results}
Given the significance of the decremental SSSP problem, it is important to understand its time complexity.

In this paper, we obtain a near-linear time algorithm for decremental \mbox{$ (1 + \epsilon) $}-approximate SSSP in weighted undirected graphs.
Its total update time is $ m^{1 + O(\log^{5/4} ((\log{n}) / \epsilon) / \log^{1/4}{n})} \log{W} $ and it maintains an estimate of the distance between the source node and every other node, guaranteeing constant worst-case query time.
The algorithm is randomized and assumes an oblivious adversary who fixes the sequence of updates in advance, an assumption that so far was also made for all other results on approximate decremental SSSP utilizing randomization.
The algorithm is always correct and the bound on its total update time holds in expectation, which makes it a Las Vegas algorithm.
In both, the weighted and the unweighted setting, our algorithm significantly improves upon previous algorithms, leaving room for running time improvements only with respect to subpolynomial factors, which so far has only been achieved in the very dense regime~\cite{BernsteinC16,Bernstein17}.
%

As a consequence of our techniques we also obtain an algorithm for the all-pairs shortest paths (APSP) problem.
For every integer $ k \geq 2 $ and every $ 0 < \epsilon \leq 1 $, we obtain a randomized decremental $ ((2 + \epsilon)^k - 1) $-approximate APSP algorithm with query time $ O (k^k) $ and total update time $ m^{1 + 1/k + O(\log^{5/4} ((\log{n}) / \epsilon) / \log^{1/4}{n}) } \log^2{W} $ in expectation.
We remark that for $ k = 2 $ and $ 1 / \epsilon = \polylog{n} $ our result gives a $ (3 + \epsilon) $-approximation with constant query time and total update time $ m^{1 + 1/2 + o(1)} \log{W} $.
For very sparse graphs with $ m = \Theta(n) $, this is almost optimal in the sense that it almost matches the static running time~\cite{ThorupZ05} of $ O (m \sqrt{n}) $, providing stretch of $ 3 + \epsilon $ instead of $ 3 $ as in the static setting.
Our result on approximate APSP has to be compared with the following prior work.
For weighted directed graphs Bernstein~\cite{Bernstein16} gave a randomized decremental $ (1 + \epsilon) $-approximate APSP algorithm with constant query time and total update time $ \tilde O (m n \log{W}) $.
For unweighted undirected graphs there are two previous results that improve upon this update time at the cost of larger approximation error.
First, for any integer $ k \geq 2 $, Bernstein and Roditty~\cite{BernsteinR11} gave a randomized decremental $ (2 k - 1 + \epsilon) $-approximate APSP algorithm with constant query time and total update time $ m n^{1/k + O(1/\sqrt{\log n}))} $.
Second, for any integer \mbox{$ k \geq 2 $}, Abraham et al.~\cite{AbrahamCT14} gave a randomized decremental $ 2^{O(\rho k)} $-approximate APSP algorithm for unweighted undirected graphs with query time $ O (k \rho) $ and total update time $ \tilde O (m n^{1/k}) $, where $ \rho = (1 + \lceil (\log{n^{1-1/k}}) / \log{(m / n^{1-1/k})} \rceil) $.

\paragraph{Outline}
We give preliminaries on decremental approximate shortest path algorithms in \Cref{sec:prelim} and provide a technical overview of our approach in \Cref{sec:technical overview}.
Our algorithm, presented in \Cref{sec:from_sssp_to_balls,sec:from_balls_to_sssp,sec:putting_together}, uses the following hierarchical approach:
Given a decremental approximate SSSP algorithm for distances up to $ D_i $ with total update time $ m^{1 + o(1)} $, we can maintain so-called approximate balls for distances up to $ D_i $ with total time $ m^{1 + o(1)} $ as well.
And given a decremental algorithm for maintaining approximate balls for distances up to $ D_i $ with total update time $ m^{1 + o(1)} $ we can use the approximate balls to define a hop set which allows us to maintain approximate shortest paths for distances up to $ D_{i+1} = n^{o(1)} D_i $ with total update time $ m^{1 + o(1)} $.
This scheme is repeated until $ D_i $ is large enough to cover the full distance range.
We have formulated the two parts of this scheme as reductions.
In \Cref{sec:from_sssp_to_balls} we give a decremental algorithm for maintaining approximate balls that internally uses a decremental approximate SSSP algorithm.
In \Cref{sec:from_balls_to_sssp} we give a decremental approximate SSSP algorithm that internally uses a decremental algorithm for maintaining approximate balls.
In \Cref{sec:putting_together} we explain the hierarchical approach for putting these two parts together and obtain the decremental $ (1 + \epsilon) $-approximate SSSP algorithm with a total update time of $ m^{1 + o(1)} $ for the full distance range.
In addition to this result, the algorithm for maintaining approximate balls, together with a suitable query algorithm, gives us a decremental approximate APSP algorithm.
This algorithm is also given in \Cref{sec:putting_together}.
Finally, we conclude the paper in \Cref{sec:conclusion}.

%% file: prelim.tex
\section{Preliminaries}\label{sec:prelim}

In this paper we want to maintain approximate shortest paths in an undirected graph $ G = (V, E) $ with positive integer edge weights in the range from $ 1 $ to $ W $, for some parameter $ W $.
The graph undergoes a sequence of \emph{updates}, which might be edge deletions or edge weight increases.
This is called the \emph{decremental setting}.
We denote by $ V $ the set of nodes of $ G $ and by $ E $ the set of edges of $ G $.
We denote by $ n $ the number of nodes of $ G $ and by $ m $ the number of edges of $ G $ before the first edge deletion.

For every weighted undirected graph $ G $, we denote the weight of an edge $ (u, v) $ in $ G $ by $ w_G (u, v) $.
The \emph{distance} $ \dist_G (u, v) $ between a node $ u $ and a node $ v $ in $ G $ is the weight of the shortest path, i.e., the minimum-weight path, between $ u $ and $ v $ in $ G $.
If there is no path between $ u $ and $ v $ in $ G $, we set $ \dist_G (x, y) = \infty $.
For every set of nodes $ U \subseteq V $ we denote by $ E [U] $ the set of edges incident to the nodes of $ U $, i.e., $ E [U] = \{ (u, v) \in E \mid u \in U \} $.\footnote{Since $ G $ is an undirected graph, this definition is equivalent to $ E [U] = \{ (u, v) \in E \mid u \in U \text{ or } v \in U \} $.}
Furthermore, for every set of nodes $ U \subseteq V $, we denote by $ G|U $ the subgraph of $ G $ induced by the nodes in $ U $, i.e., $ G|U $ contains all edges $ (u, v) $ such that $ (u, v) $ is contained in $ E $ and $ u $ and $ v $ are both contained in $ U $, or in short: $ G|U = (U, E \cap U^2) $.
Similarly, for every set of edges $ F \subseteq V^2 $ and every set of nodes $ U \subseteq V $ we denote by $ F|U $ the subset of $ F $ induced by $ U $.

We say that a distance estimate $ \delta (u, v) $ is an $ (\alpha, \beta) $-approximation of the true distance $ \dist_G (u, v) $ if $ \dist_G(u, v) \leq \delta (u, v) \leq \alpha \dist_G (u, v) + \beta $, i.e., $ \delta (u, v) $ never underestimates the true distance and overestimates it with a multiplicative error of at most $ \alpha $ and an additive error of at most $ \beta $.
If there is no additive error, we simply say $ \alpha $-approximation instead of $ (\alpha, 0) $-approximation.

In our algorithms we will use graphs that do not only undergo edge deletions and edge weight increases, but also edge insertions.
For such a graph $ H $, we denote by $ \mathcal{E} (H) $ the number of edges ever contained in~$ H $, i.e., the number of edges contained in $ H $ before any deletion or insertion plus the number of inserted edges.
We denote by $ \mathcal{W} (H) $ the number of updates to the edges in $ H $.
Similarly, for a set of edges $ F $, we denote by $ \mathcal{E} (F) $ the number of edges ever contained in $ F $ and by $ \mathcal{W} (F) $ the number of updates to the edges in $ F $.

The central data structure in decremental algorithms for exact and approximate shortest paths is the Even-Shiloach tree (short: ES-tree).
This data structure maintains a shortest paths tree from a root node up to a given depth $ D $.
\begin{lemma}[\cite{EvenS81,HenzingerK95,King99}]\label{lem:ES-tree}
There is a data structure called ES-tree that, given a weighted directed graph $ G $ undergoing deletions and edge weight increases, a root node $ s $, and a depth parameter $ D $, maintains, for every node $ v $ a value $ \delta (v, s) $ such that $ \delta (v, s) = \dist_G (v, s) $ if $ \dist_G (v, s) \leq D $ and $ \delta (v, s) = \infty $ if $ \dist_G (v, s) > D $.
It has constant query time and a total update time of $ O (mD + n) $.
\end{lemma}
Note that Dinitz, as part of his max-flow algorithm~\cite{Dinic70}, earlier developed an algorithm with similar guarantees for the decremental single-source single-sink shortest path problem~\cite{Dinitz06}.

Recent approaches for solving approximate decremental SSSP and APSP use special graphs called \emph{emulators}.
An $ (\alpha, \beta) $-emulator $ H $ of a graph $ G $ is a graph containing the nodes of $ G $ such that $ \dist_G (u, v) \leq \dist_H (u, v) \leq \alpha \dist_G (u, v) + \beta $ for all nodes $ u $ and~$ v $.\footnote{For the related notion of a spanner we additionally have to require that $ H $ is a subgraph of $ G $.}
Maintaining exact distances on~$ H $ provides an $ (\alpha, \beta) $-approximation of distances in~$ G $.
As good emulators are sparser than the original graph this is usually more efficient than maintaining exact distances on~$ G $.
However, the edges of $ H $ also have to be maintained while $ G $ undergoes updates.
For unweighted, undirected graphs undergoing edge deletions, the emulator of Thorup and Zwick (based on the second spanner construction in~\cite{ThorupZ06}), which provides a relatively good approximation, can be maintained quite efficiently~\cite{BernsteinR11}.
However the definition of this emulator requires the occasional insertion of edges into the emulator.
Thus, it is not possible to run a purely decremental algorithm on top of it.

There have been approaches to design algorithms that mimic the behavior of the classic ES-tree when run on an emulator that undergoes insertions.
The first approach by Bernstein and Roditty~\cite{BernsteinR11} extends the ES-tree to a fully dynamic algorithm and analyzes the additional work incurred by the insertions.
The second approach was introduced by us in~\cite{HenzingerKNSICOMP16} and is called \emph{monotone ES-tree}.
It basically ignores insertions of edges into $ H $ and never decreases the distance estimate it maintains.
However, this algorithm does not provide an $ (\alpha, \beta) $-approximation on \emph{any} arbitrary $ (\alpha, \beta) $-approximate emulator as it needs to exploit structural properties of the emulator to guarantee the approximation.
In~\cite{HenzingerKNSICOMP16} we gave an analysis of the monotone ES-tree when run on a specific $ (1 + \epsilon, 2) $-emulator and in the current paper we use a different analysis for our new algorithms.
If we want to use the monotone ES-tree to maintain $ (\alpha, \beta) $-approximate distances up to depth $ D $ we will set the maximum level in the monotone ES-tree to $ L = \alpha D + \beta $.
The running time of the monotone ES-tree as analyzed in~\cite{HenzingerKNSICOMP16} is as follows.
\begin{lemma}\label{lem:running_time_monotone_ES_tree}
For every $ L \geq 1 $, the total update time of a monotone ES-tree up to maximum level $ L $ on a graph $ H $ undergoing edge deletions, edge insertions, and edge weight increases is $ O ({\mathcal{E} (H) \cdot L + \mathcal{W} (H)} + n) $.
\end{lemma}

%% file: overview.tex
\section{Technical Overview}\label{sec:technical overview}

In the following we explain the main ideas of this paper, which lead to an algorithm for maintaining a hop set of a graph undergoing edge deletions.

\paragraph{General Idea}
With the well-known algorithm of Even and Shiloach we can maintain a shortest paths tree from a source node up to a given depth $ D $ under edge deletions in time $ O (m D) $.
In unweighted graphs, all simple paths have at most~$ n - 1 $ edges and therefore we can set $ D = n $ to maintain a full shortest paths tree.
In weighted graphs with positive integer edge weights from $ 1 $ to~$ W $, all simple paths have weight at most~$ (n - 1) W $ and therefore we can set $ D = nW $ to maintain a full shortest paths tree.
Using an established rounding technique~\cite{KleinS97,Cohen98,Zwick02,Bernstein09,Madry10,Bernstein16,Nanongkai14}, one can use this algorithm to maintain $ (1 +\epsilon) $-approximate single-source shortest paths up to $ h $ edges in time $ O (m h \log{(n W)} / \epsilon) $.
By setting $ h = n $, we can use this algorithm to maintain a full approximate shortest paths tree, even in weighted graphs.
This algorithm would be very efficient if the graph had a small hop diameter, i.e., if for any pair of nodes there were a shortest path with a small number of edges.
Our idea is to artificially construct such a graph.

To this end we will use a so-called \emph{hop set}.
An \emph{$ (h, \epsilon) $- hop set} $ F $ of a graph $ G = (V, E) $ is a set of weighted edges $ F \subseteq V^2 $, where the weight of each edge $ (u, v) \in F $ is at least $ \dist_G (u, v) $, such that in the graph $ H = (V, E \cup F) $ there exists, for every pair of nodes $ u $ and $ v $, a path from $ u $ to $ v $ of weight at most $ (1 + \epsilon) \dist_G (u, v) $ and with at most $ h $ hops.
In this terminology, the number of hops of a path is its number of edges.
If we run the approximate SSSP algorithm on $ H $, we obtain a running time of $ O ((m + |F|) h \log{(nW)} / \epsilon) $.
In our algorithm we will obtain an $ (n^{o(1)}, \epsilon) $-hop set of size $ m^{1 + o(1)} $ and thus, given the hop set, the running time will be $ m^{1 + o(1)} \log{(nW)} / \epsilon $.
It is however not enough to simply construct the hop set at the beginning.
We also need a dynamic algorithm for maintaining the hop set under edge deletions in $ G $.
We will present an algorithm that performs this task also in almost linear time over all deletions.

Roughly speaking, we achieve the following.
Given a graph $ G = (V, E) $ undergoing edge deletions, we can maintain a restricted hop set $ F $ such that, for all pairs of nodes $ u $ and $ v $, if the shortest path~$ \pi $ from $ u $ to~$ v $ in~$ G $ has $ h \geq n^{1/q} $ hops, then in the shortcut graph $ H = (V, E \cup F) $ there is a path from $ u $ to $ v $ of weight at most $ (1 + \epsilon) \dist_G (u, v) $ and with at most $ \lceil h / n^{1/q} \rceil \log{n} $ hops.
Our high-level idea for maintaining an (unrestricted) $ (n^{o(1)}, \epsilon) $ hop set is the following hierarchical approach.
We start with $ H_0 = G $ to maintain a hop set $ F_1 $ of $ G $, which reduces the number of hops by a factor of $ \log{n} / n^{1/q} $ at the cost of a multiplicative error of $ 1 + \epsilon $.
Given $ F_1 $, we use the shortcut graph $ H_1 = (V, E \cup F_1) $ to maintain a hop set $ F_2 $ of $ G $ that reduces the number of hops by another factor of $ \log{n} / n^{1/q} $ introducing another error of $ 1 + \epsilon $.
By repeating this process $ q $ times we arrive at a hop set that guarantees, for all pairs of nodes $ u $ and $ v $, a path of weight at most $ (1 + \epsilon)^q \dist_G (u, v) $ and with at most $ (\log{n})^{q} $ hops.
\Cref{fig:hierarchical_approach} visualizes this hierarchical approach.

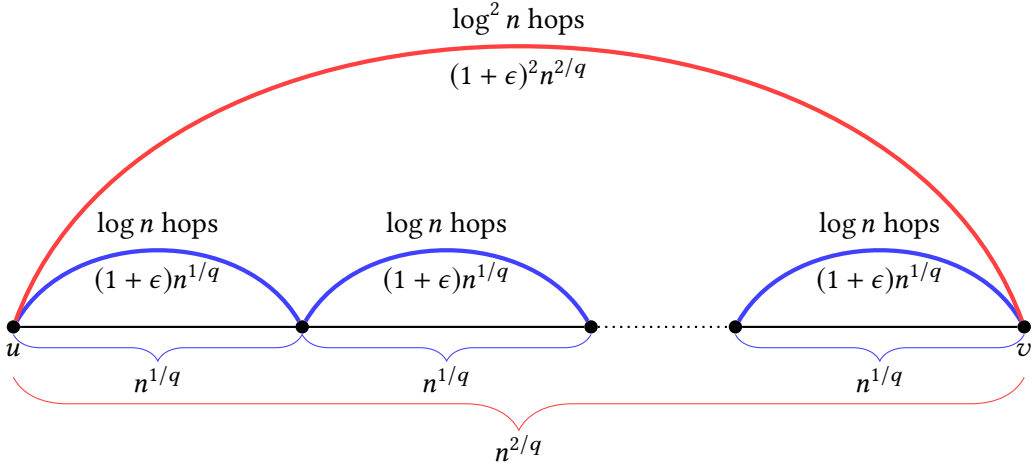
\begin{figure}[htbp!]
\centering
\begin{tikzpicture}[scale=0.95]
\tikzstyle{vertex}=[circle,fill=black,minimum size=5pt,inner sep=0pt,outer sep=0pt]
\tikzstyle{shortest-path} = [draw,thick]
\tikzstyle{shortest-path-omitted} = [draw,thick,dotted]
\tikzstyle{hop-path-small} = [draw,ultra thick,-,color=blue!75]
\tikzstyle{hop-path-big} = [draw,ultra thick,-,color=red!75]

\node[vertex,label=below:$u$] (v1) at (0, 0) {};
\node[vertex] (v2) at (4, 0) {};
\node[vertex] (v3) at (8, 0) {};
\node[vertex] (v4) at (10, 0) {};
\node[vertex,label=below:$v$] (v5) at (14, 0) {};

\path[shortest-path] (v1) -- (v3);
\path[shortest-path-omitted] (v3) -- (v4);
\path[shortest-path] (v4) -- (v5);

\draw[hop-path-small] (v1) to [bend left=60] node[black,below] {$(1+\epsilon) n^{1/q}$} node[black,above] {$ \log{n} $ hops} (v2);
\draw[hop-path-small] (v2) to [bend left=60] node[black,below] {$(1+\epsilon) n^{1/q}$} node[black,above] {$ \log{n} $ hops} (v3);
\draw[hop-path-small] (v4) to [bend left=60] node[black,below] {$(1+\epsilon) n^{1/q}$} node[black,above] {$ \log{n} $ hops} (v5);

\draw[hop-path-big] (v1) to [bend left=70] node[black,below] {$(1+\epsilon)^2 n^{2/q}$} node[black,above] {$ \log^2{n} $ hops} (v5);

\draw [blue!75,decorate,decoration={brace,amplitude=10pt}] ([yshift=-3pt]v2.center) -- node[black,below,yshift=-8pt] {$n^{1/q}$} ([yshift=-3pt]v1.center);
\draw [blue!75,decorate,decoration={brace,amplitude=10pt}] ([yshift=-3pt]v3.center) -- node[black,below,yshift=-8pt] {$n^{1/q}$} ([yshift=-3pt]v2.center);
\draw [blue!75,decorate,decoration={brace,amplitude=10pt}] ([yshift=-3pt]v5.center) -- node[black,below,yshift=-8pt] {$n^{1/q}$} ([yshift=-3pt]v4.center);

\draw [red!75,decorate,decoration={brace,amplitude=20pt}] ([yshift=-20pt]v5.center) -- node[black,below,yshift=-18pt] {$n^{2/q}$} ([yshift=-20pt]v1.center);

\end{tikzpicture}
\caption{Illustration of the hierarchical approach for maintaining the hop set reduction.
Here $ q = \Theta(\sqrt{\log{n}}) $ and $ u $ and $ v $ are nodes that are at distance $ n^{2/q} $ from each other.
First, we find a hop set that shortcuts all subpaths of weight $ n^{1/q} $ by paths of weight at most $ (1+\epsilon) n^{1/q} $ and with at most $ \log{n} $ hops.
Second, we use the shortcuts of the first hop set to find a hop set that shortcuts the path from $ u $ to $ v $ of weight $ n^{2/q} $ by a path of weight at most $ (1+\epsilon)^2 n^{2/q} $ and with at most $ \log^2{n} $ hops.
}\label{fig:hierarchical_approach}
\end{figure}

The notion of hop set was first introduced by Cohen~\cite{Cohen00} in the PRAM literature and is conceptually related to the notion of emulator.
It is also related to the notion of {\em shortest-path diameter} used in distributed computing (e.g., \cite{KhanKMPT12,Nanongkai14}).
To the best of our knowledge, the only place that this hop set concept was used before in the dynamic algorithms literature (without the name being mentioned) is Bernstein's fully dynamic $(2 + \epsilon)$-approximate APSP algorithm~\cite{Bernstein09}.
There, Bernstein shows that the clustering of Thorup and Zwick~\cite{ThorupZ05} yields an $ (n^{o(1)}, \epsilon) $-hop set by connecting each node with all nodes in its cluster.
In his fully dynamic algorithm, this clustering is recomputed \emph{from scratch} after every edge update.
Conceptually, our hop set is a decremental variant of Bernstein's hop set based however on a slightly simpler clustering.
After the preliminary version of our work appeared, Elkin and Neiman~\cite{ElkinN17} and Huang and Pettie~\cite{HuangP17a} proved that the hop set based on the Thorup-Zwick clustering provides a close-to optimal trade-off between hop parameter and size~\cite{AbboudBP17} for $ (1 + \epsilon) $-approximate distances.

\paragraph{Static Hop Set}
We first assume that $ G = (V, E) $ is an unweighted undirected graph and for simplicity we also assume that $ \epsilon $ is a constant.
We explain how to obtain a hop set of $ G $ using a randomized construction of Thorup and Zwick~\cite{ThorupZ06} based on the notion of balls of nodes.
We describe this construction and the hop-set analysis in the following.

Let $ 2 \leq p \leq \log{n} $ be a parameter and consider a sequence of sets of nodes $ A_0, A_1, \ldots, A_p $ obtained as follows.
We set $ A_0 = V $ and $ A_p = \emptyset $ and for $ 1 \leq i \leq p-1 $ we obtain the set $ A_i $ by picking each node of $ V $ independently with probability $ 1 / n^{i/p} $.
The expected size of $ A_i $ is $ n^{1 - i/p} $.
For every node~$ u $ we define the priority of~$ u $ as the maximum~$ i $ such that $ u \in A_i $.
For a node $ u $ of priority $ i $ we define
\begin{equation}
\ball (u) = \{ v \in V \mid \dist_G (u, v) < \dist_G (u, A_{i+1}) \} \label{eq:definition of ball}
\end{equation}
where $ \dist_G (u, A_{i+1}) = \min_{v \in A_{i+1}} \dist_G (u, v) $.
Note that $ \dist_G (u, A_p) = \infty $ and thus if $ u \in A_{p-1} $, then $ \ball (u) = V $.
For each node $ u $ of priority $ i $ the size of $ \ball (u) $ is $ n^{(i+1)/p} $ in expectation by the following argument:
Order the nodes in non-decreasing distance from $ u $.
Each of these nodes belongs to $ A_{i+1} $ with probability $ 1 / n^{(i+1)/p} $ and therefore, in expectation, we need to see $ n^{(i+1)/p} $ nodes until one of them is contained in $ A_{i+1} $.
It follows that the expected size of all balls of priority~$ i $ is at most $ n^{1 + 1/p} $ (the expected size of $ A_i $ times the expected size of $ \ball (u) $ for each node $ u $ of priority $ i $) and the expected size of all balls, i.e., $ \sum_{u \in V} | \ball (u) | $, is at most $ p n^{1 + 1/p} $.

Let $ F $ be the set of edges $ F = \{ (u, v) \in V^2 \mid v \in \ball (u) \} $ and give each edge $ (u, v) \in F $ the weight $ w_F (u, v) = \dist_G (u, v) $.
By the argument above, the expected size of $ F $ is at most $ p n^{1 + 1/p} $.
An argument of Thorup and Zwick~\cite{ThorupZ06} shows that the weighted graph $ H = (V, F) $ has the following property for every pair of nodes $ u $ and $ v $ and any $ 0 < \epsilon \leq 1 $ such that $ 1 / \epsilon $ is integer:\footnote{The requirement that $ 1 / \epsilon $ must be integer is not needed in the paper of Thorup and Zwick; we have added it here to simplify the exposition.}
\begin{equation*}
\dist_G (u, v) \leq \dist_H (u, v) \leq (1 + \epsilon) \dist_G (u, v) + 2 \left( 2 + \frac{2}{\epsilon} \right)^{p-2} \, .
\end{equation*}
Note that the choice of $ \epsilon $ gives a trade-off in the error between the multiplicative part $ (1 + \epsilon) $ and the additive part $ 2 (2 + 2/\epsilon)^{p-2} $.
In the literature, such a graph $ H $ is known as an \emph{emulator} of $ G $ with multiplicative error $ (1 + \epsilon) $ and additive error $ 2 ( 2 + 2 / \epsilon )^{p-2} $.\footnote{In their paper, Thorup and Zwick~\cite{ThorupZ06} actually define a graph $ H' $ whose set of edges is the union of the shortest paths trees from every node $ u $ to all nodes in its ball. This graph has the same approximation error and the same size as $ H $; since $ H' $ is a subgraph of $ G $ it is called a \emph{spanner} of $ G $.}
Roughly speaking, the strategy in their proof is as follows.
Let $ u' $ be the node following~$ u $ on the shortest path~$ \pi $ from $ u $ to $ v $ in~$ G $.
If the edge $ (u, u') $ is also contained in $ H $, then we can shorten the distance to $ v $ by $ 1 $ without introducing any approximation error (recall that we assume that $ G $ is unweighted).
Otherwise, one can show that there is a path~$ \pi' $ with at most $ p $ edges in $ H $ from $ u $ to a node $ v' $ closer to $ v $ than $ u $ such that the ratio between the weight of $ \pi' $ and the distance from $ u $ to $ v' $ is at most $ (1 + \epsilon) $, and, if $ v' = v $, then the weight of $ \pi' $ is at most $ 2 (2 + 2/\epsilon)^{p-2} $.
The proof needs the following property of the balls: for every node $ x $ of priority $ i $ and every node $ y $, either $ y \in \ball (x) $ or there is some node $ z $ of priority $ j > i $ such that $ \dist_G (x, z) = \dist_G (x, A_{i+1}) \leq \dist_G (x, y) $.
We illustrate the proof strategy in \Cref{fig:hop_set_strategy}.

\begin{figure}[htbp!]
\centering
\begin{tikzpicture}
\tikzstyle{vertex}=[circle,fill=black,minimum size=5pt,inner sep=0pt,outer sep=0pt]
\tikzstyle{shortest-path} = [draw,thick,dotted]
\tikzstyle{hop-edge} = [draw,ultra thick,-,color=blue!75]
\tikzstyle{hop-edge-missing} = [draw,ultra thick,dashed,color=blue!75]

\node[vertex,label=below:$u_0$] (u0) at (0, 0) {};
\node[vertex,label=below:$u_1~~$] (u1) at (-0.5, 1) {};
\node[vertex,label=below:$u_2~~$] (u2) at (-1.5, 3) {};
\node[vertex,label=below:$v_0$] (x0) at (1, 0) {};
\node[vertex,label=below:$v_1$] (x1) at (3, 0) {};
\node[vertex,label=below:$v_2$] (x2) at (7, 0) {};
\node[vertex,label=below:$v$] (v) at (10, 0) {};

\path[shortest-path] (u0) -- (v);

\path[hop-edge] (u0) -- (u1);
\path[hop-edge] (u1) -- (u2);
\path[hop-edge] (u2) -- (x2);
\path[hop-edge-missing] (u0) -- (x0);
\path[hop-edge-missing] (u1) -- (x1);

\path[draw,ultra thick,color=black!50,->] (0,-1) -- node[below,color=black]{decreasing distance to $ v $} (10,-1) ;
\path[draw,ultra thick,color=black!50,->] (-2,0) -- node[above,color=black,rotate=90]{increasing priority} (-2,3) ;
\end{tikzpicture}
\caption{Illustration of the approximation argument for $ p = 3 $ priorities. The dotted line is the shortest path $ \pi $ from $ u_0 $ to $ v $ in $ G $. The thick, blue edges are the edges of $ F $ used to shorten the distance to~$ v $. The dashed, blue edges are not contained in $ F $ and imply the existence of edges to nearby nodes of increasing priority. Starting from $ u_0 $, a node of priority $ 0 $, we let $ v_0 $ be the node on~$ \pi $ such that $ \dist_G (u_0, v_0) = r_0 := 1 $, i.e., the neighbor of $ u_0 $ on~$ \pi $. If the edge $ (u_0, v_0) $ is not contained in $ F $, then $ F $ contains an edge $ (u_0, u_1) $ to a node $ u_1 $ of priority at least~$ 1 $ such that $ \dist_G (u_0, u_1) \leq r_0 $. Let $ v_1 $ be the node on~$ \pi $ such that $ \dist_G (u_1, v_1) = r_1 := 1 + 2/\epsilon $. If the edge $ (u_1, v_1) $ is not contained in $ F $, then $ F $ contains an edge $ (u_1, u_2) $ to a node $ u_2 $ of priority at least $ 2 $ such that $ \dist_G (u_1, u_2) \leq r_1 $. Let $ v_2 $ be the node on~$ \pi $ such that $ \dist_G (u_2, v_2) = r_2 := (1 + 2/\epsilon) (2 + 2/\epsilon) $. Since $ 2 $ is the highest priority, $ u_2 $ contains the edge $ (u_2, v_2) $. Note that the weight of these three edges from $ F $ is at most $ r_0 + r_1 + r_2 $ and $ \dist_G (u_0, v_2) \geq r_2 - (r_0 + r_1) $. Since $ r_2 = (1 + 2/\epsilon) (r_0 + r_1) $, the ratio between these two quantities is $ (1 + \epsilon) $.}\label{fig:hop_set_strategy}
\end{figure}
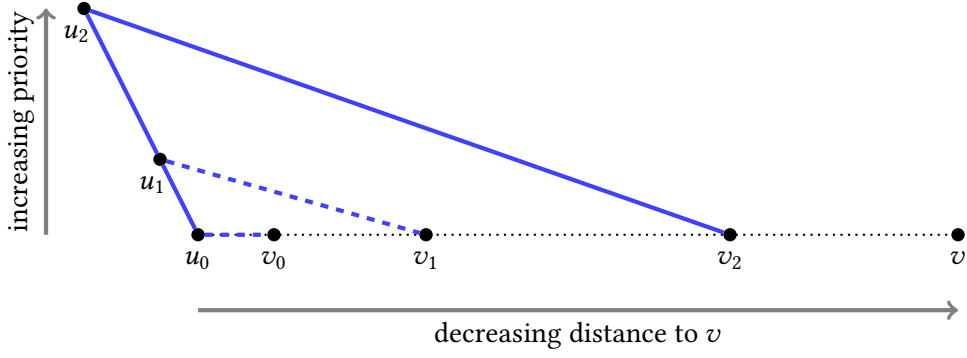

Observe that the same strategy can be used for the following hop-reduction argument:
Given any integer $ \Delta \leq n $, let $ u' $ be the node that is at distance $ \Delta $ from $ u $ on the shortest path from $ u $ to~$ v $ in~$ G $.
If the edge $ (u, u') $ is contained in $ H $, then we can shorten the distance to~$ v $ by $ \Delta $ without introducing any approximation error.
Otherwise, one can show that there is a path $ \pi' $ with at most $ p $ edges in $ H $ from $ u $ to a node $ v' $ closer to $ v $ than $ u $ such that the ratio between the weight of $ \pi' $ and the distance from $ u $ to $ v' $ is at most $ (1 + \epsilon) $, and, if $ v' = v $, then the weight of $ \pi' $ is at most $ 2 (2 + 2/\epsilon)^{p-2} \cdot \Delta $.
Every time we repeat this argument the distance to $ v $ is shortened by at least $ \Delta $.
Therefore there is a path from $ u $ to $ v $ in $ H $ with at most $ p \lceil \dist_G (u, v) / \Delta \rceil $ edges that has weight at most $ (1 + \epsilon) \dist_G (u, v) + 2 ( 2 + 2 / \epsilon )^{p-2} \cdot \Delta $.
Bernstein~\cite{Bernstein09} observed that in this type of argument the latter statement would also be true if we had removed all edges from $ F $ of weight more than $ (1 + 2/\epsilon) (2 + 2/\epsilon)^{p-2} $, which is the maximum weight of the edge to $ v' $ in the proof strategy above outlined in~\Cref{fig:hop_set_strategy}.
We will need this fact in the dynamic algorithm as it allows us to limit the depth of the balls.

By a suitable choice of $ p = \Theta (\sqrt{\log{n}}) $ (as a function of $ n $ and $ \epsilon $) we can guarantee that $ { 2 ( 2 + 2 / \epsilon )^{p-2} } \leq \epsilon n^{1/p} $ and $ n^{1/p} = n^{o(1)} $.
Now define $ q = p $ and $ \Delta_k = n^{k/q} $ for each $ 0 \leq k \leq q-2 $.
Then we have, for every $ 0 \leq k \leq q-2 $ and all pairs of nodes $ u $ and $ v $
\begin{align*}
\dist_G (u, v) \leq \dist_H (u, v) &\leq (1 + \epsilon) \dist_G (u, v) + 2 \left( 2 + \frac{2}{\epsilon} \right)^{p-2} \cdot \Delta_k \\
 &\leq (1 + \epsilon) \dist_G (u, v) + \epsilon n^{1/p} \cdot \Delta_k \\
 &= (1 + \epsilon) \dist_G (u, v) + \epsilon \Delta_{k+1} \, .
\end{align*}
Thus, if $ \Delta_{k+1} \leq \dist_G (u, v) \leq \Delta_{k+2} $, then there is a path from $ u $ to $ v $ in $ H $ of weight at most 
\begin{equation*}
(1 + \epsilon) \dist_G (u, v) + \epsilon \Delta_{k+1} \leq (1 + \epsilon) \dist_G (u, v) + \epsilon \dist_G (u, v) = (1 + 2 \epsilon) \dist_G (u, v)
\end{equation*}
and with at most $ p \lceil \dist_G (u, v) / \Delta_k \rceil \leq (p + 1) \Delta_{k+2} / \Delta_k = (p + 1) n^{2/q} = n^{o(1)} $ edges.
On the other hand, if $ \dist_G (u, v) \leq \Delta_1 = n^{1/q} $, then, as each edge of $ G $ has weight at least~$ 1 $, there is a path from $ u $ to $ v $ of weight $ \dist_G (u, v) $ with at most $ n^{1/p} = n^{o(1)} $ edges.
It follows that $ F $ is an $ (n^{o(1)}, 2 \epsilon) $-hop set of size $ O (p n^{1+1/p}) = n^{1 + o(1)} $.
Using the parameter $ \epsilon' = \epsilon/2 $ instead of $ \epsilon $, we obtain an $ (n^{o(1)}, \epsilon) $-hop set of size $ n^{1 + o(1)} $.

\paragraph{Efficient Construction}
So far we have ignored the running time for computing the balls and thus constructing $ F $, even in the static setting.
Thorup and Zwick~\cite{ThorupZ06} have remarked that a naive algorithm for computing the balls takes time $ O (m n) $.
We can reduce this running time by sampling edges instead of nodes.

We modify the process for obtaining the sequence of sets $ A_0, A_1, \ldots, A_p $ as follows.
We set $ A_0 = V $ and $ A_p = \emptyset $ and for $ 1 \leq i \leq p-1 $ we obtain the set $ A_i $ by picking each edge of $ E $ independently with probability $ 1 / m^{i/p} $ and adding both endpoints of each sampled edge to $ A_i $.
The priority of a node~$ u $ is the maximum $ i $ such that $ u \in A_i $.
We define, for every node~$ u $ of priority $ i $, $ \ball (u) $ just like in \Cref{eq:definition of ball}, but using the new definition of $ A_i $.
Note that the expected size of $ A_i $ is $ O (m^{1 - i/p}) $ for every $ 1 \leq i \leq p-1 $.

The balls can now be computed as follows.
First, following Thorup and Zwick~\cite{ThorupZ05}, we compute, for each $ 1 \leq i \leq p-1 $, $ \dist_G (u, A_i) = \min_{v \in A_i} \dist_G (u, v) $ for every node $ u $ by adding an artificial source node $ s_i $ that is connected to every node in $ A_i $ by an edge of weight $ 0 $.
Using Dijkstra's algorithm, this takes time $ O (p (m + n \log{n})) $.
Second, we compute for every node $ u $ of priority $ i $ a shortest paths tree \emph{up to depth} $ \dist_G (u, A_{i+1}) - 1 $ to obtain all nodes contained in $ \ball (u) $.
Using an implementation of Dijkstra's algorithm that only puts nodes into its queue upon their first visit this takes time $ O ( | E[ \ball (u)] | \log{n} ) $ where $ E[ \ball (u)] = \{ (x, y) \in E \mid x \in \ball (u) \text{ or } y \in \ball (u) \} $ is the set of edges incident to $ \ball (u) $.
Using the same ordering argument as before, our random sampling process for the edges guarantees that the expected size of $ E[ \ball (u)] $ is $ m^{(i+1)/p} $.
For $ 0 \leq i \leq p-1 $ the expected size of $ A_i $ is $ O (m^{1-i/p}) $ and thus these Dijkstra computations take time $ O (m^{1 + 1/p} \log{n}) $ for all nodes of priority $ i $.
By choosing $ p = \Theta(\sqrt{\log{n}})$ as described above we have $ m^{1/p} = m^{o(1)} $ and thus the balls can be computed in time $ m^{1+o(1)} $.

We define $ F $ as the set of edges $ F = \{ (u, v) \in V^2 \mid v \in \ball (u) \} $ and give each edge $ (u, v) \in F $ the weight $ w_F (u, v) = \dist_G (u, v) $.
The distance-preserving and hop-reducing properties of $ F $ still hold as stated above and its expected size is $ O (p m^{1 + 1/p}) $.
Note that $ F $ is not necessarily a sparsification of $ G $ anymore (as the bound on its size is even more than~$ m $).
For our purposes the sparsification aspect is not relevant, we only need the hop reduction.
Thus in the static setting, we can compute an $ (m^{o(1)}, \epsilon) $-hop set (which is also an $ (\epsilon, n^{o(1)}) $-hop set) of expected size $ m^{1+o(1)} $ in expected time $ m^{1+o(1)} $.

\paragraph{Maintaining Balls Under Edge Deletions}
As the graph $ G $ undergoes deletions the hop set has to be updated as well.
Unfortunately, we do not know how to maintain the balls efficiently.
However we can maintain for all nodes $ u $ the approximate ball
\begin{equation*}
\ball (u, D) = \{ v \in V \mid \log{\dist_G (u, v)} < \lfloor \log{\dist_G (u, A_{i+1})} \rfloor \text{ and } \dist_G (u, v) \leq D \}
\end{equation*}
(where $ i $ is the priority of $ u $) in time $ O (p m^{1 + 1/p} D \log{D}) $.
Note that $ \ball (u, D) $ differs from the definition of $ \ball (u) $ in the following ways.
First, we use the inequality $ \log{\dist_G (u, v)} < \lfloor \log{\dist_G (u, A_{i+1})} \rfloor $ instead of the inequality $ \dist_G (u, v) < \dist_G (u, A_{i+1}) $.
This relaxed inequality alone increases the additive error in the hop-reduction argument from $ 2 (2 + 2/\epsilon)^{p-2} \Delta $ to $ 4 (3 + 4/\epsilon)^{p-2} \Delta $ since before, a node $ v $ of higher priority than a given node $ u $ could be found directly at the boundary of the ball of $ u $, whereas now $ v $ could be twice as far away.
The increase in the additive error can easily be compensated by reducing the number of priorities $ p $ by a constant factor.
Second, we limit the balls to a certain depth~$ D $.
By using a small value of $ D $ we will only obtain a restricted hop set that provides sufficient hop reduction for nodes that are relatively close to each other.
We will show later that this is enough for our purposes.
Despite these modifications we clearly have $ \ball (u, D) \subseteq \ball (u) $ and therefore all size bounds still apply.

In the first part of the algorithm for maintaining the balls, we maintain $ \dist_G (u, A_i) $ up to threshold~$ D $ for every $ 1 \leq i \leq p-1 $ and every node $ u $.
We do this by adding an artificial source node $ s_i $ that has an edge of weight $ 0 $ to every node in $ A_i $ and maintain an ES-tree up to depth~$ D $ from $ s_i $.
This step takes time $ O (p m D) $.

Now, for every node $ u $ of priority $ i $ we maintain $ \ball (u, D) $ as follows.
We maintain an ES-tree up to depth
\begin{equation*}
\min \left( 2^{\lfloor \log{\dist_G (u, A_{i+1})} \rfloor} - 1, D \right)
\end{equation*}
and every time $ 2^{\lfloor \log{\dist_G (u, A_{i+1})} \rfloor} $ increases, we restart the ES-tree.
Naively, we incur a cost of $ O (m D) $ for each instance of the ES-tree.
However we can easily implement the ES-tree in such a way that it never processes edges that are not contained in $ E [\ball (u, D)] $.\footnote{If we prefer to use the ES-tree as a ``black box'' we can, in a preprocessing step, find the initial set $ \ball (u, D) $ and only build an ES-tree for this ball. All other nodes will never be contained in $ \ball (u, D) $ anymore as long as the value of $ 2^{\lfloor \log{\dist_G (u, A_{i+1})} \rfloor} $ remains unchanged and therefore we can remove them. This can be done in time $ O (|E [\ball (u, D)] | \log{n}) $ by using an implementation of Dijkstra's algorithm that only puts nodes into its queue upon their first visit.}
Thus, the cost of each instance of the ES-tree is $ O (| E [\ball (u, D)] | D) $.
Remember that $ \ball (u, D) \subseteq \ball (u) $ and that $ E[ \ball(u) ] $ is at most $ m^{(i+1)/p} $ in expectation.
As $ 2^{\lfloor \log{\dist_G (u, A_{i+1})} \rfloor} $ can increase at most $ \log{D} $ times until it exceeds $ D $, we initialize at most $ \log{D} $ ES-trees for the node $ u $.
Therefore the total time needed for maintaining $ \ball (u, D) $ is $ O (m^{(i+1)/p} D \log{D}) $ in expectation.
As there are at most $ O (m^{1-i/p}) $ nodes of priority $ i $ in expectation, the total time needed for maintaining all approximate balls is $ O (p m^{1 + 1/p} D \log{D}) $ in expectation.

\paragraph{Decremental Approximate SSSP}
Let us first sketch an algorithm for maintaining shortest paths from a source node $ s $ with a running time of $ m^{1 + 1/2 + o(1)} $ for which we use $ p = \Theta (\sqrt{\log{n}}) $ priorities.
We set $ \Delta = \lfloor \sqrt{n} \rfloor $, $ p $ such that $ (2 + 4/\epsilon) (3 + 4/\epsilon)^{p-2} \leq \epsilon n^{1/p} $ and $ n^{1/p} = n^{o(1)} $, and $ D = \lceil \epsilon n^{1/p} \rceil $.
We maintain single-source shortest paths up to depth $ D $ from $ s $ using the ES-tree, which takes time $ O (m D) = m n^{1/2 + o(1)} $.
To maintain approximate shortest paths to nodes that are at distance more than~$ D $ from~$ s $ we use the following approach.
We maintain $ \ball (u, D) $ for every node $ u $, as sketched above, which takes time $ O (p m^{1 + 1/p} D \log{D}) = m^{1 + 1/2 + o(1)} $ in expectation.
At any time, we set the hop set to be the set of edges $ F = \{ (u, v) \in V^2 \mid v \in \ball (u, D) \} $ and give each edge $ (u, v) \in F $ the weight $ w_F (u, v) = \dist_G (u, v) $.
By our arguments above, the weighted graph $ H = (V, F) $ has the following property: for every pair of nodes $ u $ and $ v $ such that $ \dist_G (u, v) \geq D $ (where $ D \geq n^{1/p} \Delta $) there is a path $ \pi' $ in $ H $ of weight at most $ (1 + \epsilon) \dist_G (u, v) + \epsilon n^{1/p} \Delta \leq (1 + 2 \epsilon) \dist_G (u, v) $ and with at most $ p \lceil \dist_G (u, v) / \Delta \rceil $ edges.

To maintain approximate shortest paths for nodes at distance more than $ D $ from $ s $ we will now use the hop reduction in combination with the following rounding technique.
We set $ \varphi = \epsilon \Delta / (p + 1) $ and let $ H' $ be the graph resulting from rounding up every edge weight in $ H $ to the next multiple of~$ \varphi $.
By using $ H' $ instead of $ H $ we incur an error of $ \varphi $ for every edge on the approximate shortest path~$ \pi' $.
Thus in $ H' $, $ \pi' $ has weight at most
\begin{align*}
(1 + 2 \epsilon) \dist_G (u, v) +  \lceil p \dist_G (u, v) / \Delta \rceil \cdot \varphi &= (1 + 2 \epsilon) \dist_G (u, v) + \epsilon \dist_G (u, v) \\
 &\leq (1 + 3 \epsilon) \dist_G (u, v) \, .
\end{align*}
The efficiency now comes from the observation that we can run the algorithm on the graph $ H'' $ in which every edge weight in $ H' $ is scaled down by a factor of $ 1 / \varphi $.
The graph $ H'' $ has integer weights and the weights of all paths in $ H' $ and $ H'' $ differ exactly by the factor $ 1 / \varphi $.
Thus, instead of maintaining a shortest paths tree up to depth $ n $ in $ H $ we only need to maintain a shortest paths tree in $ H'' $ up to depth $ n / \varphi = p \sqrt{n} / \epsilon $.
In this way we obtain a $ (1 + 3 \epsilon) $-approximation for all nodes such that $ \dist_G (u, v) \geq D $.

However, we cannot simply use the ES-tree on $ H'' $ because as edges are deleted from~$ G $, nodes might join the approximate balls and therefore edges might be inserted into $ F $ and thus into $ H'' $.
This means that a dynamic shortest paths algorithm running on $ H'' $ would not be situated in a purely decremental setting.
However the insertions have a ``nice'' structure.
We can deal with them by using a previously developed technique, called \emph{monotone ES-tree}~\cite{HenzingerKNSICOMP16}.
The main idea of the monotone ES-tree is to ignore the level decreases made possible by inserting edges.
The hop-set proof still goes through, even though we are not arguing about the current distance in $ H'' $ anymore, but the level of a node $ u $ in the monotone ES-tree.
Maintaining the monotone ES-tree for distances up to $ D $ in $ H'' $ takes time $ O (\mathcal{E} (H'') D) $ where $ \mathcal{E} (H'') $ is the number of edges ever contained in~$ H'' $ (including edges that are inserted over time) and $ D = O (n^{1/2 + 1/p}) $ as explained above.
Each insertion of an edge into $ F $ corresponds to a node joining $ \ball (u, D) $ for some node $ u $.
For a fixed node $ u $ of priority $ i $ there are at most $ \log{D} $ possibilities for nodes to join $ \ball (u, D) $ (namely each time $ \lfloor \log{\dist_G (u, A_{i+1})} \rfloor $ increases) and every time at most $ m^{(i+1)/p} $ nodes will join in expectation.
It follows that $ \mathcal{E} (H) $ is $ m^{1 + o(1)} $ in expectation and the running time of this step is $ m^{1 + 1/2 + o(1)} $ in expectation.

The almost linear-time algorithm is just slightly more complicated.
Here we use $ p = \Theta (\sqrt{\log{n}}) $ priorities and a hierarchy of $ q = \sqrt{p} $ hop reductions.
We further set $ \Delta_k = n^{k/q} $ for each $ 0 \leq k \leq q-2 $.
In the algorithm we will maintain, for each $ 0 \leq k \leq q-2 $ a hop set $ F_k $ such that for every pair of nodes $ u $ and $ v $ with $ \Delta_{k+1} \leq \dist_G (u, v) \leq \Delta_{k+2} $ there is a path from $ u $ to $ v $ in $ H_k = (V, F_k) $ of weight at most $ (1 + 2 \epsilon) \dist_G (u, v) $ and with at most $ p \dist_G (u, v) / \Delta_k \leq p n^{2/q} $ hops.
To achieve this we use the following hierarchical approach.
Given the hop set $ F_k $ we can maintain approximate shortest paths up to depth $ \Delta_{k+2} $ in time $ m^{1 + o(1)} $ and given a data structure for maintaining approximate shortest paths up to depth $ \Delta_k $ we can maintain approximate balls and thus the hop set $ F_{k+1} $ in time $ m^{1 + o(1)} $.
The hierarchy ``starts'' with using the ES-tree as an algorithm for maintaining an (exact) shortest paths tree up to depth $ n^{2/q} $.
Thus, running efficient monotone ES-trees on top of the hop sets and maintaining the hop sets (using efficient monotone ES-trees) go hand in hand.

There are two obstacles in implementing this hierarchical approach when we want to maintain the approximate balls in each of the $ q $ layers of the hierarchy.
First, in our algorithm for maintaining the approximate balls sketched above we have used the ES-tree as an exact decremental SSSP algorithm.
In the hierarchical approach we have to replace the ES-tree with the monotone ES-tree which only provides approximate distance estimates.
This will lead to approximation errors that increase with the number of layers.
Second, by the arguments above the number of edges in $ F_k $ is $ O (m^{1 + 1/p}) $ for each $ 0 \leq k \leq q-2 $.
In the algorithm for maintaining the approximate balls for the next layer, this bound however is not good enough because we run a separate instance of the monotone ES-tree for each node $ u $.
We deal with this issue by running the monotone ES-tree in the subgraph of $ G $ induced by the nodes initially contained in $ \ball (u) $.
For a node $ u $ of priority~$ i $ this subgraph contains $ m_i = m^{(i+1)/p} $ edges in expectation and we can recursively run our algorithm on this smaller graph.
By this process we incur a factor of $ m^{1/p} $ in the running time each time we increase the depth of the recursion.
This results in a total update time of $ m^{1 + O(q/p)} $ which is $ m^{1 + O(1/q)} = m^{1 + o(1)} $ since $ q = \sqrt{p} $.

\paragraph{Extension to Weighted Graphs}
The hop set construction described above only works for unweighted graphs.
However, the main property that we needed was $ \dist_G (u, v) \leq n $ for any pair of nodes $ u $ and $ v $.
Using the rounding technique mentioned above, we can construct for each $ 0 \leq i \leq \lfloor \log{nW} \rfloor $ a graph $ G_i $ such that for all pairs of nodes $ u $ and~$ v $ with $ 2^i \leq \dist_G (u, v) \leq 2^{i+1} $ we have $ \dist_{G_i} (u, v) \leq 4 n / \epsilon $ and the shortest path in $ G_i $ can be turned into a $ (1+\epsilon) $-approximate shortest path in $ G $ by scaling up the edge weights.
We now run $ O (\log{(nW)}) $ instances of our algorithm, one for each graph $ G_i $, and maintain the hop set and approximate SSSP for each of them.

We only need to refine the analysis of the hop-set property in the following way.
Remember that in the analysis we considered the shortest path $ \pi $ from $ u $ to $ v $ and defined the node $ u' $ that is at distance $ \Delta $ from $ u $ on $ \pi $.
If the hop set contained the edge $ (u, u') $ we could reduce the distance to~$ v $ by $ \Delta $.
In weighted graphs (even after the scaling), we cannot guarantee there is a node at distance exactly $ \Delta $ from $ u $ on $ \pi $.
Therefore we define $ u' $ as the furthest node that is at distance \emph{at most} $ \Delta $ from $ u $ on $ \pi $.
Furthermore we define $ u'' $ as the neighbor of $ u' $ on $ \pi $, i.e., $ u'' $ is at distance \emph{at least} $ \Delta $ from $ u $.
Now if the hop set contains the edge $ (u, u') $ we first use the edge $ (u, u') $ from the hop set, and then the edge $ (u', u'') $ from the original graph to reduce the distance to~$ v $ by at least~$ \Delta $ with only $ 2 $~hops.
Note that for unweighted graphs it was sufficient to only use the edges of the hop set.
For weighted graphs we really have to add the edges of the hop set to the original graph in our algorithm.

%% file: sssp_to_balls.tex
\section{From Approximate SSSP to Approximate Balls}\label{sec:maintaining_emulator}\label{sec:from_sssp_to_balls}

In the following we show how to maintain the approximate balls of every node if we already have an algorithm for maintaining approximate shortest paths.
In our reduction we will use the algorithm for maintaining approximate shortest paths as a ``black box'', requiring only very few properties.

We can view the balls as a distance oracle with exponentially increasing stretch.
Similar to other distance oracles, we assign integer values called priorities to the nodes and ensure that the balls have the following structural property: for every pair of nodes $ u $ and $ v $, either $ v $ is in the ball of $ u $, or there is some node $ v' $ close to $ u $ that has higher priority than $ u $.
In the first case, we have found an estimate of the distance between $ u $ and $ v $ as the approximate shortest path algorithm is used to maintain an estimate of the distance between $ u $ and all nodes in its ball.
In the second case, we repeat the process for the nodes $ v' $ and $ u $, incurring the `detour' of going from~$ u $ to~$ v' $ first.
As the number of priorities is limited, this strategy succeeds eventually.
In our analysis, we explicitly bound the distance between $ u $ and $ v' $ (and thus the weight of the `detour') by a function $ s (x, l) $, where $ x $ is the distance between $ u $ and $ v $ and $ l $ is the difference in priorities between $ u $ and $ v' $.
In \Cref{sec:putting_together} it will become clear why our bound on $ s (x, l) $ is good enough for our purposes of using the balls as a hop set and as a distance oracle, respectively.
In addition to this structural property, we need to bound the total size of the balls to obtain useful applications.
This bound can be ensured by an appropriate randomized assignment of priorities, with some complications arising from the fact that the black box decremental SSSP algorithm does not provide exact distances.
This also helps for bounding the total update time for dynamically maintaining the balls.
Formally, we prove the following statement in this section.

\begin{proposition}\label{pro:from_sssp_to_balls}
Assume there is a decremental approximate SSSP algorithm \textsc{ApproxSSSP} with the following properties, using fixed values $ \alpha \geq 1 $, $ \beta \geq 0 $, and $ D \geq 1 $:
Given a weighted graph $ G = (V, E) $ undergoing edge deletions and edge weight increases with $ n $ nodes and initially $ m $ edges, and a fixed source node $ s \in V $, \textsc{ApproxSSSP} maintains, in total update time $ T (m, n) $, for every node $ v \in V $ a distance estimate $ \distest (s, v) $ such that:
\begin{enumerate}[label=\textbf{A\arabic*}]
\item $ \distest (s, v) \geq \dist_G (s, v) $ \label{prop:no_underestimation_of_distance}
\item If $ \dist_G (s, v) \leq D $, then $ \distest (s, v) \leq \alpha \dist_G (s, v) + \beta $. \label{prop:approximation_guarantee}
\item After every update in $ G $, \textsc{ApproxSSSP} returns, for every node $ v $ such that $ \distest (s, v) $ has changed, $ v $ together with the new value of $ \distest (s, v) $. \label{prop:returning_updated_nodes}
\end{enumerate}

Then there is a decremental algorithm \textsc{ApproxBalls} for maintaining approximate balls with the following properties: Given a weighted graph $ G = (V, E) $ undergoing edge deletions and edge weight increases with $ n $ nodes and initially $ m $ edges, and parameters $ \p \geq 2 $ and $ 0 < \epsilon \leq 1 $, it assigns to every node $ u \in V $ a number from $ 0 $ to $ \p - 1 $, called the \emph{priority} of $ u $, and maintains for every node $ u \in V $ a set of nodes $ B (u) $ and a distance estimate $ \distest (u, v) $ for every node $ v \in B (u) $ such that:
\begin{enumerate}[label=\textbf{B\arabic*}]
\item For every node $ u $ and every node $ v \in B (u) $ we have $ \dist_G (u, v) \leq \distest (u, v) \leq \alpha \dist_G (u, v) + \beta $. \label{prop:approximation_in_balls}
\item For all $ x \geq 0 $, set $ s (x, 0) = x $, and for all $ x \geq 0 $ and $ l \geq 1 $, set \label{prop:finding_node_of_higher_priority}
\begin{equation*}
s (x, l) = a (a + 1)^{l-1} x + ((a + 1)^l - 1) b / a \, ,
\end{equation*}
where $ a = (1 + \epsilon) \alpha $ and $ b = (1 + \epsilon) \beta $.
Then for every $ 0 \leq i \leq \p - 1 $, every node~$ u $ of priority~$ i $, and every node~$ v $ such that $ s (\dist_G (u, v), \p - 1 - i) \leq D $, either (1) $ v \in B (u) $ or (2) there is some node $ v' $ of priority $ j > i $ such that $ u \in B (v') $ and $ \dist_G (u, v') \leq s (\dist_G (u, v), j - i) $.
\item In expectation, $ \sum_{u \in V} \mathcal{B} (u) = O (\p m^{1+1/\p} \log{D} / \epsilon) $, where $ \mathcal{B} (u) $ denotes the number of nodes ever contained in $ B (u) $ over the sequence of updates to $ G $. \label{prop:size_of_balls}
\item The total update time of \textsc{ApproxBalls} is \label{prop:time_for_maintaining_balls}
\begin{equation*}
t (m, n, \p, \epsilon) = O \left( \left(\p m^{1 + 1/\p} + \sum_{0 \leq i \leq \p - 1} \frac{m}{m^{i/\p}} \cdot T (m_i, n_i) \right) \cdot \log{n} \frac{\log{D}}{\epsilon} + \p \cdot T (m, n) \right)
\end{equation*}
in expectation, where, for each $ 0 \leq i \leq \p - 1 $, $ m_i = O (m^{(i+1)/\p}) $ and $ n_i = O (m^{(i+1)/\p}) $.
\item After every update in $ G $, \textsc{ApproxBalls} returns all pairs of nodes $ u $ and $ v $ such that $ v $ joins $ B (u) $, $ v $ leaves $ B (u) $, or $ \hat{\distest} (u, v) $ changes. \label{prop:returning_updated_nodes_of_balls}
\end{enumerate}
\end{proposition}

Note that by our definition of $ s (x, l) $ we have $ s (x, 1) = a x + b $ for all $ x \geq 0 $ and $ s (x, l + 1) = (a + 1) s (x, l) + b $ for all $ x \geq 0 $ and $ l \geq 1 $.

Our algorithm for maintaining the approximate balls $ B (u) $ for every node $ u \in V $ is as follows:
\begin{enumerate}
\item At the initialization we set $ F_0 = E $ and $ F_\p = \emptyset $ and for $ 1 \leq i \leq \p - 1 $, a set of edges $ F_i $ is obtained from sampling each edge of $ E $ independently with probability $ 1 / m^{i/\p} $.
For every $ 0 \leq i \leq \p - 1 $ we set $ A_i = \{ v \in V \mid \exists (v, w) \in F_i \} $ and for every node $ v \in V $, we set the \emph{priority of $ u $} to be the maximum $ i $ such that $ v \in A_i $.

\item For each $ 1 \leq i \leq \p - 1 $ we run an instance of \textsc{ApproxSSSP} from an artificial source node $ s_i $ that has an edge of weight $ 0 $ to every node in $ A_i $.
We denote by $ \distest (u, A_i) $ the distance estimate provided by \textsc{ApproxSSSP} and set $ \distest (u, A_\p) = \infty $ for every node $ u \in V $.

\item For every $ 0 \leq i \leq \p - 1 $ and every node $ u \in V $ of priority~$ i $, we maintain the value
\begin{equation*}
r (u) = \min \left (\frac{(1 + \epsilon)^{\lfloor \log_{1 + \epsilon}{\distest (u, A_{i+1})} \rfloor} - \beta}{\alpha}, D + 1 \right)
\end{equation*}
and at the initialization and each time $ r (u) $ increases we do the following:
\begin{enumerate}
\item Compute the set of nodes $ R (u) = \{ v \in V \mid \dist_G(u, v) < r (u) \} $.
\item Run an instance of \textsc{ApproxSSSP} from $ u $ in $ G | R (u) $, the subgraph of $ G $ induced by $ R (u) $.
Let $ \distest (u, v) $ denote the estimate of the distance between $ u $ and $ v $ in $ G | R (u) $ maintained by \textsc{ApproxSSSP}.
\item Maintain $ B (u) = \{ v \in V \mid \distest (u, v) \leq \alpha D + \beta \} $: every time $ \distest (u, v) $ changes for some node~$ v $ we check whether the inequality $ \distest (u, v) \leq \alpha D + \beta $ still holds, and if not we remove $ v $ from $ B (u) $.
\end{enumerate}
\end{enumerate}

Note that \textsc{ApproxBalls} has Property~\ref{prop:returning_updated_nodes_of_balls}, i.e., it returns changes in the approximate balls and the distance estimates, which is possible because \textsc{ApproxSSSP} has Property~\ref{prop:returning_updated_nodes}.

\subsection{Relation to Exact Balls}

In the following we compare the approximate balls maintained by our algorithm to the exact balls, as used by Thorup and Zwick~\cite{ThorupZ06}.
We show how the main properties of exact balls translate to approximate balls.
We use the following definition of the (exact) ball of a node $ u $ of priority~$ i $:
\begin{equation*}
\ball (u) = \{ v \in V \mid \dist_G (u, v) < \dist_G (u, A_{i+1}) \} \, .
\end{equation*}
The balls have the following simple property: If $ v \notin \ball (u) $, then there is a node $ v' $ of priority $ j > i $ such that $ \dist_G (u, v') \leq \dist_G (u, v) $.
We show that a relaxed version of this statement also holds for the approximate balls.

\begin{lemma}\label{lem:next_node_in_chain}
Let $ 0 \leq i \leq \p - 1 $, let $ u $ be a node of priority~$ i $, and let $ v $ be a node such that $ \dist_G (u, v) \leq D $.
If $ v \notin B (u) $, then there is a node $ v' $ of priority $ j > i $ such that $ \dist_G (u, v') \leq a \dist_G (u, v) + b $, where $ a = (1 + \epsilon) \alpha $ and $ b = (1 + \epsilon) \beta $.
\end{lemma}

\begin{proof}
We show the following:
If $ \dist_G(u, A_{i+1}) > a \dist_G(u, v) + b $, then $ v \in B (u) $.
The claim then follows from contraposition: If $ v \notin B (u) $, then $ \dist_G (u, A_{i+1}) \leq a \dist_G(u, v) + b $ and thus there exists some node $ v' \in A_{i+1} $ (of priority $ j \geq i + 1 $) such that $ \dist_G (u, v') \leq a \dist_G (u, v) + b $.

Assume that $ \dist_G (u, A_{i+1}) \geq a \dist_G (u, v) + b $.
Since $ \distest (u, A_{i+1}) \geq \dist_G (u, A_{i+1}) $, by Property~\ref{prop:no_underestimation_of_distance} we have
\begin{equation*}
\distest (u, A_{i+1}) \geq \dist_G (u, A_{i+1}) > a \dist_G (u, v) + b = (1 + \epsilon) (\alpha \dist_G (u, v) + \beta)
\end{equation*}
which is equivalent to
\begin{equation*}
\dist_G (u, v) < \frac{\frac{\distest (u, A_{i+1})}{1 + \epsilon} - \beta}{\alpha} \, .
\end{equation*}
Since
\begin{equation*}
(1 + \epsilon)^{\lfloor \log_{1 + \epsilon}{\distest (u, A_{i+1}}) \rfloor} \geq (1 + \epsilon)^{\log_{1 + \epsilon}{\distest (u, A_{i+1})} - 1} = \frac{\distest (u, A_{i+1})}{1 + \epsilon}
\end{equation*}
it follows that
\begin{equation*}
\dist_G (u, v) < \frac{(1 + \epsilon)^{\lfloor \log_{1 + \epsilon}{\distest (u, A_{i+1})} \rfloor} - \beta}{\alpha}
\end{equation*}
Since we have assumed that $ \dist_G (u, v) \leq D $, we get $ \dist_G (u, v) < r (u) $ by the definition of $ r (u) $.
The latter inequality also holds for all nodes on a shortest path $ \pi $ from $ u $ to $ v $ in $ G $, and, as distances in~$ G $ are non-decreasing, this in particular is true at the last point in time where $ r (u) $ has changed.
Therefore, all nodes of $ \pi $ are contained in $ R (u) $, which implies that $ \dist_{G|R(u)} (u, v) = \dist_G (u, v) \leq D $.
Thus, by Property~\ref{prop:approximation_guarantee}, it follows that $ \distest (u, v) \leq \alpha \dist_{G|R(u)} (u, v) + \beta \leq \alpha D + \beta $, i.e., $ v \in B (u) $, as desired.
\end{proof}

We now show that the approximate balls are contained in the exact balls.
The exact balls are useful in our analysis because we can easily bound their size.

\begin{lemma}\label{lem:approximate_ball_included_in_ball}
At any time $ B (u) \subseteq \ball (u) $ for every node $ u $.
\end{lemma}

\begin{proof}
Let $ R (u) = \{ v \in V \mid \dist_G (u, v) < r (u) \} $ denote the set of nodes at distance of at most $ r (u) $ from $ u $ at the last time $ r (u) $ has increased.
Note that $ B (u) $ is a set of nodes of the graph $ G|R(u) $ and therefore $ B (u) \subseteq R (u) $.
It remains to show that $ R (u) \subseteq \ball (u) $.

Let $ v \in R (u) $ and let $ i $ be the priority of $ u $.
If $ i = \p - 1 $, then the claim is trivially true because $ \ball (u) $ contains all nodes that are connected to $ u $ in $ G $.
Consider thus the case $ 0 \leq i < \p - 1 $.
In the case $ 0 \leq i < \p - 1 $ remember that
If $ \dist_G (u, A_{i+1}) \geq r (u) $, we trivially have $ \dist_G (u, v) < r (u) \leq \dist_G (u, A_{i+1}) $.
If on the other hand $ \dist_G (u, A_{i+1}) < r (u) $, then in particular $ \dist_G (u, A_{i+1}) \leq D $ by the definition of $ r (u) $ and by Property~\ref{prop:approximation_guarantee} we have $ \distest (u, A_{i+1}) \leq \alpha \dist_G (u, A_{i+1}) + \beta $.
It follows that
\begin{align*}
\dist_G (u, v) < r (u) &\leq \frac{(1 + \epsilon)^{\lfloor \log_{1 + \epsilon}{\distest (u, A_{i+1})} \rfloor} - \beta}{\alpha} \\
&\leq \frac{(1 + \epsilon)^{\log_{1 + \epsilon}{\distest (u, A_{i+1})}} - \beta}{\alpha}
= \frac{\distest (u, A_{i+1}) - \beta}{\alpha} \leq \dist_G (u, A_{i+1}) \, .
\end{align*}
In both cases we get $ \dist_G (u, v) < \dist_G (u, A_{i+1}) $ and as this is the defining property of $ \ball (u) $ we have $ v \in \ball (u) $.
\end{proof}

\begin{lemma}\label{lem:size_of_ball}
At any time, for every $ 0 \leq i \leq \p - 1 $ and every node $ u $ of priority~$ i $, we have $ | \ball (u) | = O (m^{(i+1)/\p}) $ and $ | E [\ball (u)] | = O (m^{(i+1)/\p}) $ in expectation over all random choices in sampling the set $ F_{i+1} $.
\end{lemma}

\begin{proof}
The claim is trivially true if $ \ball (u) = \{ u \} $ and we thus assume that $ \ball (u) \supset \{ u \} $ in the following.
It further suffices to prove that $ | E [\ball (u)] | = O (m^{(i+1)/\p}) $ as $ | \ball (u) | \leq 2 | E [\ball (u)] | $ and, as the claim immediately holds for $ i = \p - 1 $, we only need to consider the case $ i < \p - 1  $.
For every edge $ e = (v, w) \in E $ we define $ \dist_G (u, e) = \min (\dist_G (u, v), \dist_G (u, w)) $.
Order the edges of the graph according distance from $ u $ under to this definition of $ \dist_G (u, e) $ for each edge $ e $, where ties are broken in an arbitrary but fixed order.
Let $ e = (u, v) $ be the first edge of $ F_{i+1} $ in this order and let $ E' \subseteq F_i $ be the set of edges that are strictly smaller than $ e $ in this order.
As each edge of the graph is contained in $ F_{i+1} $ with probability $ 1 / m^{(i+1)/\p} $ independently, we have $ |E'| \leq m^{(i+1)/\p} $ in expectation.

Assume without loss of generality that $ \dist_G (u, v) \leq \dist_G (u, w) $, i.e., $ \dist_G (u, e) = \dist_G (u, v) $.
Let $ v' \in \ball (u) $ and let $ e' = (v', w') $ be some edge incident to $ v' $; such an edge must exist because $ \ball (u) \supset \{ u \} $.
Since the node $ v $ is contained in $ A_{i+1} $ we have
\begin{equation*}
\dist_G (u, e') \leq \dist_G (u, v') < \dist_G (u, A_{i+1}) \leq \dist_G (u, v) = \dist_G (u, e)
\end{equation*}
which implies $ e' \in E' $.
It follows that $ E [\ball (u)] \subseteq E' $ and thus $ | E [\ball (u)] | \leq | E' | \leq m^{(i+1)/\p} $ in expectation, as desired.
\end{proof}

\subsection{Properties of Approximate Balls}

We now show that the approximate balls and the corresponding distance estimates have Properties \ref{prop:approximation_in_balls}--\ref{prop:time_for_maintaining_balls}.
We first show that the distance estimates for nodes in the approximate balls have the desired approximation guarantee, although they have been computed in subgraphs of $ G $.

\begin{lemma}[Property~\ref{prop:approximation_in_balls}]\label{lem:approximation_guarantee_in_ball}
For every pair of nodes $ u $ and $ v $ such that $ v \in B (u) $ we have $ \dist_G (u, v) \leq \distest (u, v) \leq \alpha \dist_G (u, v) + \beta $.
\end{lemma}

\begin{proof}
By Property~\ref{prop:no_underestimation_of_distance} we have $ \distest (u, v) \geq \dist_{G|R(u)} (u, v) $ and since $ G|R(u) $ is a subgraph of $ G $ we have $ \dist_{G|R(u)} (u, v) \geq \dist_G (u, v) $.
Therefore the inequality $ \distest (u, v) \geq \dist_G (u, v) $ follows.

Since $ v \in B (u) $ we have $ \distest (u, v) \leq \alpha D + \beta $.
If $ \dist_G (u, v) \geq D $, then trivially $ \distest (u, v) \leq \alpha D + \beta \leq \alpha \dist_G (u, v) + \beta $.
If $ \dist_G (u, v) < D $, then there is a path $ \pi $ from $ u $ to $ v $ in $ G $ of weight at most $ D $.
This path was also contained in previous versions of $ G $, possibly with smaller weight.
In particular, $ \pi $ was also contained in the version of $ G $ at the last point in time for which the set $ R (u) $ was recomputed.
Since $ v \in B (u) \subseteq R (u) $ we therefore also have $ v' \in R (u) $ for every node $ v' $ on $ \pi $.
It follows that $ \pi $ is contained in $ G|R(u) $ and thus $ \dist_{G|R(u)} (u, v) = \dist_G (u, v) \leq D $.
By Property~\ref{prop:approximation_guarantee} we then have $ \distest (u, v) \leq \alpha \dist_{G|R(u)} (u, v) + \beta = \alpha \dist_G (u, v) + \beta $.
\end{proof}

We show now that the approximate balls have a certain structural property that either allows us to shortcut the path between two nodes or helps us in finding a nearby node of higher priority.
\begin{lemma}[Property~\ref{prop:finding_node_of_higher_priority}]\label{lem:approximation_property}
For all $ x \geq 0 $, set $ s (x, 0) = x $, and for all $ x \geq 0 $ and $ l \geq 1 $, set
\begin{equation*}
s (x, l) = a (a + 1)^{l-1} x + ((a + 1)^l - 1) b / a \, ,
\end{equation*}
where $ a = (1 + \epsilon) \alpha $ and $ b = (1 + \epsilon) \beta $.
Then for every $ 0 \leq i \leq \p - 1 $, every node~$ u $ of priority~$ i $, and every node~$ v $ such that $ s (\dist_G (u, v), \p - 1 - i) \leq D $, either (1) $ v \in B (u) $ or (2) there is some node~$ v' $ of priority $ j > i $ such that $ u \in B (v') $ and $ \dist_G (u, v') \leq s (\dist_G (u, v), j-i) $.
\end{lemma}

\begin{proof}
As noted above, by our definition of $ s (x, l) $ we have $ s (x, 1) = a x + b $ for all $ x \geq 0 $ and $ s (x, l + 1) = (a + 1) s (x, l) + b $ for all $ x \geq 0 $ and $ l \geq 1 $.
Note that since $ s (\cdot, \cdot) $ is non-decreasing in its second argument we have, for all $ 0 \leq l \leq \p - 1 - i $, $ s (\dist_G (u, v), l) \leq s (\dist_G (u, v), \p - 1 - i) \leq D $.

If $ v \in B (u) $, then we are done.
Otherwise, by \Cref{lem:next_node_in_chain}, there is some node $ v_1 $ of priority $ j_1 \geq i + 1 $ such that
\begin{equation*}
\dist_G(v_1, u) \leq a \dist_G(u, v) + b = s (\dist_G (u, v), 1) \leq D \, .
\end{equation*}
Thus, if $ u \in B (v_1) $, then we are done.
Otherwise, by \Cref{lem:next_node_in_chain}, there is some node $ v_2 $ of priority $ j_2 \geq j_1+1 \geq i + 2 $ such that
\begin{equation*}
\dist_G(v_2, v_1) \leq a \dist_G(v_1, u) + b \, .
\end{equation*}
By the triangle inequality we have
\begin{align*}
\dist_G(v_2, u) &\leq \dist_G(v_2, v_1) + \dist_G(v_1, u) \\
 &\leq a \dist_G(v_1, u) + b + \dist_G(v_1, u) \\
 &= (a + 1) \dist_G(v_1, u) + b \\
 &\leq (a + 1) s (\dist_G (u, v), 1) + b \\
 &= s (\dist_G (u, v), 2) \leq D \, .
\end{align*}
We now repeat this argument to obtain nodes $ v_1 $, $ v_2 $, \ldots $ v_l $ of priorities $ j_1 $, $ j_2 $, \ldots, $ j_l $ such that $ j_l \geq i + l $ and
\begin{equation*}
\dist_G (v_l, u) \leq s (\dist_G (u, v), l) \leq D
\end{equation*}
until $ u \in B (v_l) $.
This happens eventually since $ A_\p = \emptyset $ and thus for any node $ v_l $ of priority $ \p - 1 $ such that $ \dist_G (v_l, u) \leq s (\dist_G (u, v), l) \leq D $ the following reasoning applies:
Since $ \distest (v_l, A_\p) = \infty $, we always have $ r (v_l) = D + 1 $ and thus $ \dist_G (v_l, u) < r (v_l) $.
The latter inequality also holds for all nodes on a shortest path $ \pi $ from $ u $ to $ v $ in $ G $, and, as distances in $ G $ are non-decreasing, this in particular is true at the last point in time
where $ r (v_l) $ has changed.
Therefore, all nodes of $ \pi $ are contained in $ R (v_l) $, which implies that $ \dist_{G | R (v_l)} (v_l, u) = \dist_G (v_l, u) $.
Thus, by Property~\ref{prop:approximation_guarantee}, it follows that $ \distest (v_l, u) \leq \alpha \dist_{G | R (v_l)} (v_l, u) + \beta \leq \alpha D + \beta $, i.e., $ u \in B (v_l) $, as desired.
\end{proof}

Next, we bound the size of the system of approximate balls we maintain.
Here we use the fact that we can easily bound the size of the exact ball $ \ball (u) $ for every node~$ u $ and that by our definitions we ensure that the approximate balls are subsets of the exact balls.

\begin{lemma}[Size of Approximate Balls (Property~\ref{prop:size_of_balls})]\label{lem:size_of_approximate_balls}
In expectation, we have $ \sum_{u \in V} \mathcal{B} (u) = O (\p m^{1+1/\p} \log{D} / \epsilon) $, where $ \mathcal{B} (u) $ denotes the number of nodes ever contained in $ B (u) $ over the sequence of updates to $ G $.
\end{lemma}

\begin{proof}
We first bound $ \mathcal{B} (u) $, the number of nodes ever contained in the approximate ball $ B (u) $, of some node $ u $.
Let $ i $ denote the priority of $ u $.
Remember that nodes are joining $ B (u) $ only when $ r (u) $ increases and that
\begin{equation*}
r (u) = \min \left (\frac{(1 + \epsilon)^{\lfloor \log_{1 + \epsilon}{\distest (u, A_{i+1})} \rfloor} - \beta}{\alpha}, D + 1 \right) \, .
\end{equation*}
Thus, $ r (u) $ can only increase if $ \lfloor \log_{1 + \epsilon}{\distest (u, A_{i+1})} \rfloor $ increases and the left term in the minimum is at most $ D + 1 $.
Since $ 1 + \epsilon \geq 1 $ it follows that $ r (u) $ increases only $ O (\log_{1 + \epsilon}{D}) = O (\log{D} / \epsilon) $ times, once it has non-negative value.
As $ B (u) \subseteq \ball (u) $ by \Cref{lem:approximate_ball_included_in_ball}, after every increase of $ r (u) $ only nodes contained in $ \ball (u) $ can join~$ B (u) $.
By \Cref{lem:size_of_ball} the size of $ \ball (u) $ is $ O (m^{(i+1)/\p}) $ in expectation over all random choices in sampling the set $ F_{i+1} $.
Thus, the number of nodes ever contained in $ B (u) $ is $ \mathcal{B} (u) = O (m^{(i+1)/\p} \log{D} / \epsilon) $ in expectation.

As the number of nodes of priority~$ i $ is $ O (m / m^{i/\p}) $ in expectation over all random choices in sampling the set $ F_i $, the number of nodes ever contained in the approximate balls is $ \sum_{u \in V} \mathcal{B} (u) = O (\p m^{1+1/\p} \log{D} / \epsilon) $ in expectation.
Here we use that $ F_i $ and $ F_{i+1} $ are sampled independently.
\end{proof}

Finally, we analyze the running time of our algorithm for maintaining the approximate balls.
Since we use the data structure \textsc{ApproxSSSP} as a black box, the running time of our algorithm depends on the running time of \textsc{ApproxSSSP}.

\begin{lemma}[Running Time (Property~\ref{prop:time_for_maintaining_balls})]
The total time needed for maintaining the sets $ B (u) $ for all nodes $ u \in V $ is
\begin{equation*}
O \left( \left(\p m^{1 + 1/\p} + \sum_{0 \leq i \leq \p - 1} \frac{m}{m^{i/\p}} \cdot T (m_i, n_i) \right) \cdot \log{n} \frac{\log{D}}{\epsilon} + \p \cdot T (m, n) \right)
\end{equation*}
in expectation, where, for each $ 0 \leq i \leq \p - 1 $, $ m_i = O (m^{(i+1)/\p}) $ and $ n_i = O (m^{(i+1)/\p}) $.
\end{lemma}

\begin{proof}
The initialization in Step~1 of the algorithm, where we determine the sets $ A_0, \ldots, A_\p $ takes time $ O(\p m + n) $.
In Step~2, we run for each $ 1 \leq i \leq \p - 1 $ an instance of \textsc{ApproxSSSP} with depth~$ D $.
This takes time $ \p T (m, n) $.
Step~3, where we maintain for every $ 0 \leq i \leq \p - 1 $ and every node~$ u $ of priority~$ i $ the approximate ball and corresponding distance estimates, can be analyzed as follows.
Remember that every time $ r (u) $ increases we first compute $ R (u) $, the set of nodes that are at distance of at most $ r (u) $ from $ u $.
Using an implementation of Dijkstra's algorithm that only puts nodes into its queue upon their first visit, this takes time $ O (| E[R(u)] | \log{n}) $, where $ E[R(u)] $ is the set of edges incident to $ R(u) $.
By \Cref{lem:approximate_ball_included_in_ball} we have $ R (u) \subseteq \ball (u) $ and by \Cref{lem:size_of_ball} we have $ | \ball (u) | = O (m^{(i+1)/\p}) $ and $ | E [\ball (u)] | = O (m^{(i+1)/\p}) $ in expectation over all random choices in sampling the set $ F_{i+1} $.
Thus, computing $ R (u) $ takes time $ O (m^{(i+1)/\p} \log{n}) $ in expectation.
We then maintain an instance of \textsc{ApproxSSSP} up to depth $ D $ on $ G|R(u) $, the subgraph of~$ G $ induced by~$ R(u) $.
Note that $ G|R(u) $ has at most $ m_i = O (m^{(i+1)/\p}) $ edges and $ n_i = O (m^{(i+1)/\p}) $ nodes in expectation and therefore this takes time $ T (m_i, n_i) $ in expectation.
As $ r (u) $ increases $ O (\log{D} / \epsilon) $ times and the number of nodes of priority~$ i $ is at most $ O (m / m^{i/\p}) $ in expectation over all random choices in sampling the set $ F_i $, Step~3 takes time
\begin{equation*}
O \left( \sum_{0 \leq i \leq \p - 1} \left( T (m_i, n_i) + m^{(i+1)/\p} \log{n} \right) \cdot \frac{m}{m^{i/\p}} \frac{\log{D}}{\epsilon} \right) 
\end{equation*}
in expectation.
Now the claimed total running time claimed for all three steps follows.
\end{proof}

%% file: balls_to_sssp.tex
\section{From Approximate Balls to Approximate SSSP}\label{sec:approximation_guarantee_ES_trees_rounded}\label{sec:approximation_guarantee_ES_trees}\label{sec:from_balls_to_sssp}

In the following we show how to maintain an approximate shortest paths tree if we already have an algorithm for maintaining approximate balls.
Our main tool in this reduction is a hop set that we define from the approximate balls.
We will add the ``shortcut'' edges of the hop set to the graph and scale down the edge weights, maintaining the approximate shortest paths with a monotone ES-tree.

The main challenge in this approach is bounding the approximation guarantee of the hop set.
While this encompasses a proof of the approximation guarantee of the static variant of the hop set, such a ``static'' proof is not sufficient in the decremental setting; the monotone ES-tree needs to exploit the additional structure of the approximate balls that define the hop set.
In particular, the following structural property is useful here: for every pair of nodes $ u $ and $ v $, either $ v $ is in the ball of $ u $, or there is some node $ v' $ close to $ u $ that has higher priority than $ u $.
If the distance between $ u $ and $ v' $ is measured by some function $ s (x, l) $, where $ x $ is the distance between $ u $ and $ v' $ and $ l $~is the difference in priorities between $ v' $ and $ u $, then we can give an upper bound on the value of $ s (x, l) $ such that our approximation guarantee proof still goes through.
Concerning the running time, there mainly are two factors that affect the running time of the monotone ES-tree.
The first factor is the size of the hop set, more precisely the total number of edges ever contained in the hop set over all updates to the input graph.
This number is bounded by the by the total size of the (approximate) balls, more precisely by the total number of nodes ever contained in the balls.
The second factor is the the maximum depth considered in the tree.
The maximum depth can be kept relatively small by rounding and scaling down all edge weights.
This give some additional approximation error to account for, but by the right choice of parameters we can balance both costs at a moderate level.
Formally, we prove the following statement in this section.

\begin{proposition}\label{pro:from_balls_to_sssp}
Assume there is a decremental algorithm \textsc{ApproxBalls} for maintaining approximate balls with the following properties, using fixed values $ a \geq \alpha \geq 1 $, $ b \geq \beta \geq 0 $, and $ \hat{D} \geq 1 $.
Given a weighted graph $ G = (V, E) $ undergoing edge deletions and edge weight increases with $ n $ nodes and initially $ m $ edges, and a parameter $ k \geq 2 $, it assigns to every node $ u \in V $ a number from $ 0 $ to $ k - 1 $, called the \emph{priority} of $ u $, and maintains, in total update time $ t (m, n, k) $, for every node $ u \in V $ a set of nodes $ B (u) $ and, for every node $ v \in B (u) $, a distance estimate $ \hat{\distest} (u, v) $ such that:
\begin{enumerate}[label=\textbf{B\arabic*}]
\item For every node $ u $ and every node $ v \in B (u) $ we have $ \dist_G (u, v) \leq \hat{\distest} (u, v) \leq \alpha \dist_G (u, v) + \beta $. \label{prop:approximation_in_balls2}
\item There is a non-decreasing function $ s (\cdot, \cdot) $ such that, for all $ x \geq 0 $, $ s (x, 0) \leq x $ and $ s (x, 1) \leq a x + b $ for some $ a \geq \alpha $ and $ b \geq \beta $ and, for all $ l \geq 1 $, \label{prop:finding_node_of_higher_priority2}
\begin{equation*}
s (x, l + 1) \leq (\alpha + 1 + \epsilon) (\alpha a s (x, l) + \alpha b + \beta) + \beta \, .
\end{equation*}
guaranteeing the following:
For every $ 0 \leq i \leq \p - 1 $, every node $ u $ of priority~$ i $ and every node $ v $ such that $ s (\dist_G (u, v), \p - 1 - i) \leq \hat{D} $, either (1) $ v \in B (u) $ or (2) there exists some node $ v' \in V $ of priority $ j > i $ such that $ u \in B (v') $ and $ \dist_G (u, v') \leq s (\dist_G (u, v'), j - i) $.
\item After every update in $ G $, \textsc{ApproxBalls} returns all pairs of nodes $ u $ and $ v $ such that $ v $ joins $ B (u) $, $ v $ leaves $ B (u) $, or $ \hat{\distest} (u, v) $ changes. \label{prop:returning_updated_nodes_of_balls2}
\end{enumerate}

Then there is an approximate SSSP data structure \textsc{ApproxSSSP} with the following properties:
Given a weighted graph $ G $ undergoing edge deletions and edge weight increases with $ n $ nodes and initially $ m $ edges, a fixed source node $ s $, and parameters $ p $, $ \Delta $, $ D $, and $ \epsilon $ such that
\begin{equation*}
2 \leq p \leq \frac{\sqrt{\log{n}}}{\sqrt{\log{\left( \frac{4 a^3}{\epsilon} \right)}}} \, ,
\end{equation*}
$ \Delta \geq b $, $ n^{1/p} \Delta \leq \hat{D} $, $ D \geq \Delta $ and $ 0 < \epsilon \leq 1 $,
it maintains a distance estimate $ \distest (s, v) $ for every node $ v \in V $ such that:
\begin{enumerate}[label=\textbf{A\arabic*}]
\item $ \distest (s, v) \geq \dist_G (s, v) $ \label{prop:no_underestimation_of_distance2}
\item If $ \dist_G (s, v) \leq D $, then $ \distest (s, v) \leq (\alpha + 2 \epsilon) \dist_G (s, v) + \epsilon n^{1/p} \Delta $ \label{prop:approximation_guarantee2}
\item The total update time of \textsc{ApproxSSSP} is \label{prop:total_update_time_approx_SSSP}
\begin{equation*}
T (m, n, \Delta, D, \epsilon) = t (m, n, p) + O \left(p \left( \alpha D / \Delta + n^{1/p} \right) \left(m + \sum_{u \in V} \mathcal{B} (u) \right) / \epsilon + n \right)
\end{equation*}
where $ \mathcal{B} (u) $ denotes the number of nodes ever contained in $ B (u) $ over the sequence of updates to $ G $.
\item After every update in $ G $, \textsc{ApproxSSSP} returns each node $ v $ such that $ \distest (s, v) $ has changed together with the new value of $ \distest (s, v) $. \label{prop:returning_updated_nodes2}
\end{enumerate}
\end{proposition}

We assume without loss of generality that the distance estimate maintained by \textsc{ApproxBalls} is non-decreasing.
If \textsc{ApproxBalls} ever reports a decrease we can ignore it because then Property~\ref{prop:approximation_in_balls2} will still hold as distances in $ G $ are non-decreasing under edge deletions and edge weight increases.

\subsection{Algorithm Description}

The algorithm \textsc{ApproxSSSP} maintains the set of edges $ F = \{ (u, v) \in V^2 \mid v \in B (u) \} $ such that, for every node $ u $ and every node $ v \in B (u) $, the edge $ (u, v) $ has weight $ w_F (u, v) = \min (\hat{\distest} (u, v), \hat{\distest} (v, u)) $ if also $ u \in B (v) $ and $ w_F (u, v) = \hat{\distest} (u, v) $ otherwise.
We update $ F $ every time in the algorithm \textsc{ApproxBalls} a node joins or leaves an approximate ball or if the distance estimate $ \hat{\distest} (u, v) $ increases for some pair of nodes $ u $ and $ v $.
By Property~\ref{prop:returning_updated_nodes_of_balls2} this information is returned by \textsc{ApproxBalls} after every update in $ G $.
Thus the set of edges $ F $ undergoes insertions, deletions, and weight increases.

In the following we will define a shortcut graph $ H'' $ with scaled-down edge weights and our algorithm \textsc{ApproxSSSP} will simply run a monotone ES-tree~\cite{HenzingerKNSICOMP16} from $ s $ in $ H'' $.
The monotone ES-tree has property~\ref{prop:no_underestimation_of_distance2}, which is apparent from the pseudocode provided in Algorithm~\ref{alg:monotone_ES_tree}.
We denote the weight of an edge $ (u, v) $ in $ G $ by $ w_G (u, v) $ and define $ H $ as a graph that has the same nodes as $ G $ and contains all edges of $ G $ and $ F $ that have weight at most $ D + n^{1/p} \Delta $.
We set the weight of every edge $ (u, v) $ in $ H $ to $ w_H (u, v) = \min (w_G (u, v), w_F (u, v)) $.
We set
\begin{equation*}
\varphi = \frac{\epsilon \Delta}{p + 1}
\end{equation*}
and define $ H' $ as the graph that has the same nodes and edges as $ H $ and in which every edge $ (u, v) $ has weight
\begin{equation*}
w_{H'} (u, v) = \left\lceil \frac{w_H (u, v)}{\varphi} \right\rceil \cdot  \varphi \, ,
\end{equation*}
i.e., we round every edge weight to the next multiple of $ \varphi $.
Furthermore, we define $ H'' $ as the graph that has the same nodes and edges as $ H' $ and in which every edge $ (u, v) $ has weight
\begin{equation*}
w_{H''} (u, v) = \frac{w_{H'} (u, v)}{\varphi} = \left\lceil \frac{w_H (u, v)}{\varphi} \right\rceil \, ,
\end{equation*}
i.e., we scale down every edge weight by a factor of $ 1 / \varphi $.
We maintain a monotone ES-tree with maximum level
\begin{equation*}
L = (\alpha + 2 \epsilon) D / \varphi + (p + 1) n^{1/p}
\end{equation*}
from $ s $ and denote the level of a node $ v $ in this tree by $ \lev (v) $.
For every node $ v $ our algorithm returns the distance estimate $ \distest (s, v) = \lev (v) \cdot \varphi $.
Note that the graph $ H'' $ has integer edge weights and, as $ F $ might undergo insertions, deletions, and edge weight increases, the same type of updates might occur in $ H'' $.
Furthermore, observe that the rounding guarantees that
\begin{equation*}
w_H (u, v) \leq w_{H'} (u, v) \leq w_H (u, v) + \varphi
\end{equation*}
for every edge $ (u, v) $ of $ H' $.

\subsection{Running Time Analysis}

We first provide the running time analysis.
We run the algorithm in a graph in which we scale down the edge weights by a factor of $ \varphi $ and round up to the next integer.
This makes the algorithm efficient.

\begin{lemma}[Running Time (Property~\ref{prop:total_update_time_approx_SSSP})]
The total update time of a monotone ES-tree with maximum level $ L = (\alpha + 2 \epsilon) D / \varphi + (p + 1) n^{1/p} $ on $ H'' $ is
\begin{equation*}
O \left( p \left( \alpha D / \Delta + n^{1/p} \right) \left(m + \sum_{u \in V} \mathcal{B} (u) \right) / \epsilon + n \right)
\end{equation*}
where $ \mathcal{B} (u) $ denotes the number of nodes ever contained in $ B (u) $ over the sequence of updates to $ G $.
\end{lemma}

\begin{proof}
By \Cref{lem:running_time_monotone_ES_tree}, the total time needed for maintaining the monotone ES-tree with maximum level $ L $ on $ H'' $ is
\begin{equation*}
O (\mathcal{E} (H'') \cdot L + \mathcal{W} (H'') + n)
\end{equation*}
where $ \mathcal{E} (H'') $ is the number of edges ever contained in $ H'' $ and $ \mathcal{W} (H'') $ is the number of updates (i.e., edge deletions, edge weight increases, and edge insertions) to $ H'' $.

Remember that $ \varphi = \epsilon \Delta / (p + 1) $ and thus $ L = O (p (\alpha D / (\epsilon \Delta) + n^{1/p})) $.
We now bound $ \mathcal{E} (H'') $ and $ \mathcal{W} (H'') $.
Note that at any time $ H'' $~has the same edges as~$ H $ and each edge of $ H $ either is an edge from~$ G $, which contains $ m $ edges, or is an edge from~$ F $.
As $ F $ is defined via the approximate balls (i.e., $ (u, v) \in F $ if and only $ v \in B (u) $), $ \mathcal{E} (F) $, the number of edges ever contained in $ F $, is at most $ \sum_{u \in V} \mathcal{B} (u) $, the total number of nodes ever contained in the approximate balls.
It follows that $ \mathcal{E} (H'') \leq m + \mathcal{E} (F) \leq m + \sum_{u \in V} \mathcal{B} (u) $.
For every edge counted by $ \mathcal{E} (H'') $ we need to consider at most one insertion and at most one deletion as well as at most $ (D + n^{1/p} \Delta) / \varphi $ edge weight increases since we have limited the maximum edge weight in $ H $ to $ D  + n^{1/p} \Delta $.
Note that
\begin{equation*}
(D + n^{1/p} \Delta) / \varphi = (D + n^{1/p} \Delta) (p + 1) / (\epsilon \Delta) = O (p (D/\Delta + n^{1/p}) / \epsilon) \, .
\end{equation*}
Therefore we have
\begin{equation*}
\mathcal{W} (H'') \leq 2 \mathcal{E} (H'') + \mathcal{E} (H'') \cdot (D + n^{1/p} \Delta) / \varphi = O \left( \left( m + \sum_{u \in V} \mathcal{B} (u) \right) \cdot p (D/\Delta + n^{1/p}) / \epsilon \right) \, .
\end{equation*}
We conclude that
\begin{equation*}
\mathcal{E} (H'') \cdot L + \mathcal{W} (H'') = O \left( p (\alpha D / \Delta + n^{1/p}) \left( m + \sum_{u \in V} \mathcal{B} (u) \right) / \epsilon \right)
\end{equation*}
and thus the claimed running time follows.
\end{proof}

\subsection{Definitions of Values for Approximation Guarantee}

Before we analyze the approximation guarantee we define the most important values used in the analysis and provide bounds on their growth.
We set
\begin{equation*}
r_0 = \Delta
\end{equation*}
and for every $ 0 \leq i \leq p-1 $ we set
\begin{align*}
s_i &= a r_i + b \, , \\
w_i &= \alpha s_i + \beta \, \text{, and} \\
r_i &= \frac{(\alpha + 1 + \epsilon) \sum_{0 \leq j \leq i-1} w_j + \beta}{\epsilon} \text{~~(if $ i \geq 1 $)} \, .
\end{align*}
Intuitively, $ r_i $ is the distance by which we would like to shortcut the shortest path to the source node using a single hop-set edge for nodes of priority~$i$.
If this shortcut attempt fails, $ s_i $ is the distance at which we would like to find a nearby node of priority $ i + 1 $ and $ w_i $ is the weight of the hop-set edge to such a node.

We additionally set
\begin{equation*}
\gamma_i = (\alpha + 1 + \epsilon) \sum_{i \leq j \leq p-2} w_j + \beta
\end{equation*}
for every $ 0 \leq i \leq p-1 $, which equivalently can be obtained by setting $ \gamma_{p-1} = \beta $ and $ \gamma_i = \gamma_{i+1} + (\alpha + 1 + \epsilon) w_i $ for every $ 0 \leq i \leq p-2 $.
Finally, we set
\begin{equation*}
\gamma = \gamma_0 + 2 \epsilon \Delta \, .
\end{equation*}
Here $ \gamma_i $ is, intuitively speaking, the amount of additive error we will make on a hop-set path for a node of priority~$ i $ and $ \gamma $ captures some additional rounding error for nodes of priority $ 0 $.

\begin{lemma}\label{lem:alternative_formulation_of_radius}
For all $ 0 \leq i \leq p - 1 $, $ \epsilon r_i = \gamma_0 - \gamma_i + \beta $
\end{lemma}

\begin{proof}
Using the definition of $ \gamma_i $, for all $ 0 \leq i \leq p-1 $, we get 
\begin{equation*}
\gamma_0 - \gamma_i + \beta = (\alpha + 1 + \epsilon) \sum_{0 \leq j \leq p-2} w_j - (\alpha + 1 + \epsilon) \sum_{i \leq j \leq p-2} w_j + \beta = (\alpha + 1 + \epsilon) \sum_{0 \leq j \leq i-1} w_j + \beta = \epsilon r_i \, . \qedhere
\end{equation*}
\end{proof}

\begin{lemma}\label{lem:p_is_balancing}
$ (4 a^3 / \epsilon)^p \leq n^{1/p} $
\end{lemma}

\begin{proof}
Remember that we have
\begin{equation*}
p \leq \frac{\sqrt{\log{n}}}{\sqrt{\log{\left( \frac{4 a^3}{\epsilon} \right)}}} \, .
\end{equation*}
We only need to rewrite both expressions as follows:
\begin{gather*}
n^{1/p} = 2^{1/p \cdot \log{n}} \geq 2^{\frac{\sqrt{\log{ \left( \frac{4 a^3}{\epsilon} \right)}}}{\sqrt{\log{n}}} \cdot \log{n}} = 2^{\sqrt{\log{ \left( \frac{4 a^3}{\epsilon} \right)}} \cdot \sqrt{\log{n}}} \\
\left( \frac{4 a^3}{\epsilon} \right)^p = 2^{p \cdot \log{ \left( \frac{4 a^3}{\epsilon} \right)}} \leq 2^{\frac{\sqrt{\log{n}}}{\sqrt{\log{ \left( \frac{4 a^3}{\epsilon} \right)}}} \cdot \log{ \left( \frac{4 a^3}{\epsilon} \right)}} = 2^{\sqrt{\log{n}} \cdot \sqrt{\log{ \left( \frac{4 a^3}{\epsilon} \right)}}} \, . \qedhere
\end{gather*}
\end{proof}

\begin{lemma}\label{lem:bounds_on_radius_and_weight_sum}
For all $ 0 \leq i \leq p-1 $ we have
\begin{equation*}
\sum_{0 \leq j \leq i} w_j \leq \frac{(4^{i+1} - 1) a^{3i + 2} \Delta}{\epsilon^i} \, .
\end{equation*}
\end{lemma}

\begin{proof}
Remember that $ \epsilon \leq 1 \leq \alpha \leq a $ and $ \beta \leq b \leq \Delta $.
Now observe that for all $ 1 \leq i \leq p-1 $ we have
\begin{equation*}
r_i = \frac{(\alpha + 1 + \epsilon) \sum_{0 \leq j \leq i-1} w_j + \beta}{\epsilon} \leq \frac{3 a \sum_{0 \leq j \leq i-1} w_j + \Delta}{\epsilon}
\end{equation*}
and for all $ 0 \leq i \leq p-1 $ we have
\begin{equation*}
w_i = \alpha s_i + \beta \leq a s_i + b = a (a r_i + b) + b = a^2 r_i + a b + b \leq a^2 r_i + 2 a \Delta \, .
\end{equation*}

We now prove the inequality by induction on $ i $.
We begin with the base case $ i = 0 $ where $ r_0 = \Delta $ and
\begin{equation*}
\sum_{0 \leq j \leq 0} w_j = w_0 \leq a^2 r_0 + 2 a \Delta = a^2 \Delta + 2 a \Delta \leq 3 a^2 \Delta = \frac{(4 - 1) a^2 \Delta}{\epsilon^0} \, .
\end{equation*}
In the induction step we assume that $ i \geq 1 $:
\begin{align*}
\sum_{0 \leq j \leq i} w_j &= \sum_{0 \leq j \leq i-1} w_j + w_i \\
 &\leq \sum_{0 \leq j \leq i-1} w_j + a^2 r_i + 2 a b \\
 &\leq \sum_{0 \leq j \leq i-1} w_j + a^2 \cdot \frac{3 a \sum_{0 \leq j \leq i-1} w_j + b}{\epsilon} + 2 a b \\
 &\leq \frac{(3 a^3 + 1) \sum_{0 \leq j \leq i-1} w_j + a^2 b + 2 a b}{\epsilon} \\
 &\leq \frac{(3 a^3 + 1) (4^i - 1) a^{3(i-1) + 2} \Delta + a^2 b + 2 a b}{\epsilon^i} \\
 &\leq \frac{(3 a^3 + 1) (4^i - 1) a^{3(i-1) + 2} \Delta + a^2 \Delta + 2 a \Delta}{\epsilon^i} \\
 &\leq \frac{((3 + 1) (4^i - 1) + 3) a^{3i + 2} \Delta}{\epsilon^i} \\
 &\leq \frac{(4^{i+1} - 1) a^{3i + 2} \Delta}{\epsilon^i} \, . \qedhere
\end{align*}
\end{proof}

\begin{lemma}\label{lem:bound_on_additive_error}
$ a \gamma + b \leq \epsilon n^{1/p} \Delta $.
\end{lemma}

\begin{proof}
Remember that we have $ \epsilon \leq 1 \leq \alpha \leq a $ and $ \beta \leq b \leq \Delta $.
By \Cref{lem:bounds_on_radius_and_weight_sum} we have
\begin{equation*}
\sum_{0 \leq j \leq p-2} w_j \leq \frac{(4^{p-1} - 1) a^{3p - 4} \Delta}{\epsilon^{p-2}}
\end{equation*}
We now get:
\begin{align*}
\frac{a \gamma + b}{\epsilon} &= \frac{a \gamma_0 + 2 \epsilon a \Delta + b}{\epsilon} \\
 &= \frac{a (\alpha + 1 + \epsilon) \sum_{0 \leq j \leq p-2} w_j + a \beta + 2 \epsilon a \Delta + b}{\epsilon} \\
 &\leq \frac{a (a + 1 + \epsilon) \sum_{0 \leq j \leq p-2} w_j + a \Delta + 2 a \Delta + \Delta}{\epsilon} \\
 &\leq \frac{3 a^2 \sum_{0 \leq j \leq p-2} w_j + 4 a \Delta}{\epsilon} \\
 &\leq \frac{3 a^2 (4^{p-1} - 1) a^{3p - 4} \Delta + 4 a \Delta}{\epsilon^{p-1}} \\
 &\leq \frac{4^p a^{3 p} \Delta}{\epsilon^p} \\
 &= (4 a^3 / \epsilon)^p \Delta \leq n^{1/p} \Delta \, .
\end{align*}
The last inequality follows from \Cref{lem:p_is_balancing}.
\end{proof}

\begin{lemma}\label{lem:bound_on_radius}
$ a r_{p-1} + b \leq n^{1/p} \Delta $.
\end{lemma}

\begin{proof}
By the definitions of $ r_{p-1} $ and $ \gamma_0 $ we have $ r_{p-1} = \gamma_0 / \epsilon $.
Since $ \gamma_0 \leq \gamma $ and $ a \gamma + b \leq \epsilon n^{1/p} \Delta $ by \Cref{lem:bound_on_additive_error}, we have
\begin{equation*}
a r_{p-1} + b = a \frac{\gamma_0}{\epsilon} + b \leq \frac{a \gamma_0 + b}{\epsilon} \leq \frac{a \gamma + b}{\epsilon} \leq n^{1/p} \Delta \, . \qedhere
\end{equation*}
\end{proof}

\begin{lemma}\label{lem:auxiliary_inequality_distance_to_next_node}
For all $ 0 \leq i \leq j \leq p - 1 $, $ s (r_i, j - i) \leq r_j $
\end{lemma}

\begin{proof}
Fix some $ 0 \leq i \leq p-2 $.
The proof is by induction on $ j $.
In the first base case $ j = i $, the claim is trivially true as $ s (r_i, 0) \leq r_i $.
Now remember that for $ j \geq 1 $ we have
\begin{equation*}
r_j = \frac{(\alpha + 1 + \epsilon) \sum_{0 \leq j' \leq j-1} w_{j'} + \beta}{\epsilon} \geq (\alpha + 1 + \epsilon) w_{j-1} + \beta = (\alpha + 1 + \epsilon) (\alpha s_{j-1} + \beta) + \beta \, .
\end{equation*}
Thus, in the second base case $ j = i + 1 $ the claim holds because $ s (r_i, j - i) = s (r_i, 1) \leq a r_i + b = s_i = s_{j - 1} \leq r_j $.
Finally, consider the induction step where we assume that the inequality holds for $ j - 1 $ and have to show that it also holds for $ j $, where $ j \geq i + 2 $.
By the induction hypothesis we have $ s (r_i, j - 1 - i) \leq r_{j - 1} $ and since $ j - i \geq 2 $ we have
\begin{equation*}
s (r_i, j - i) \leq (\alpha + 1 + \epsilon) (\alpha a s (r_i, j - i - 1) + \alpha b + \beta) + \beta \, .
\end{equation*}
We now get:
\begin{align*}
r_j &\geq (\alpha + 1 + \epsilon) (\alpha s_{j-1} + \beta) + \beta \\
 &= (\alpha + 1 + \epsilon) (\alpha a r_{j-1} + \alpha b + \beta) + \beta \\
 &\geq (\alpha + 1 + \epsilon) (\alpha a s (r_i, j - i - 1) + \alpha b + \beta) + \beta \\
 &\geq s (r_i, j - i) \, . \qedhere
\end{align*}
\end{proof}

\subsection{Analysis of Approximation Guarantee}

We now analyze the approximation error of a monotone ES-tree maintained on $ H'' $.
This approximation error consists of two parts.
The first part is an approximation error that comes from the fact that the monotone ES-tree only considers paths from $ s $ with a relatively small number of edges and therefore has to use edges from the hop set~$ F $.
The second part is the approximation error we get from rounding the edge weights.
We first give a formula for the approximation error that depends on the priority of the nodes and their distance to the root of the monotone ES-tree.

Before we give the proof we review a few properties of the monotone ES-tree (see~\cite{HenzingerKNSICOMP16} for the full algorithm).
Similar to the classic ES-tree, the monotone ES-tree with root $ s $ maintains a level $ \lev (v) $ for every node $ v $.
The monotone ES-tree is initialized by computing a shortest paths tree up to depth $ L $ from $ s $ in $ H'' $ and thus, initially, $ \lev (v) = \dist_{H''} (s, v) $.
A single deletion or edge weight increase in~$ G $ might result in a sequence of deletions, weight increases and insertions in~$ F $, and thus~$ H'' $. 
The monotone ES-tree first processes the insertions and then the deletions and edge weight increases.
It handles deletions and edge weights increases in the same way as the classic ES-tree.
Once the level $ \lev (u) $ of a node $ u $ exceeds the maximum level $ L $, we set $ \lev (u) = \infty $.
The procedure for handling the insertion of an edge $ (u, v) $ is trivial: it only stores the new edge and in particular does \emph{not} change $ \lev (u) $ or $ \lev (v) $.
For completeness we list the pseudocode of the monotone ES-tree in Algorithm~\ref{alg:monotone_ES_tree}.

\begin{algorithm}
\caption{Monotone ES-tree}
\label{alg:monotone_ES_tree}

\SetKwProg{procedure}{Procedure}{}{}
\SetKwFunction{initialize}{Initialize}
\SetKwFunction{delete}{Delete}
\SetKwFunction{increase}{Increase}
\SetKwFunction{insert}{Insert}
\SetKwFunction{updateLevels}{UpdateLevels}

\tcp{Internal data structures:}
\tcp{$ N (u) $: for every node $ u $ a heap $ N (u) $ whose intended use is to store for every neighbor $ v $ of $ u $ in the current graph the value of $ \lev (v) + w_{H''} (u, v) $}
\tcp{$ Q $: global heap whose intended use is to store nodes whose levels might need to be updated}

\BlankLine

\procedure{\initialize{}}{
	Compute shortest paths tree from $ s $ in $ H'' $ up to depth $ L $\;
	\ForEach{$ u \in V $}{
		Set $ \lev (u) = \dist_{H''} (s, u) $\;
		\lFor{every edge $ (u, v) $ in $ H'' $}{
			insert $ v $ into heap $ N(u) $ of $ u $ with key $ \lev(v) + w_{H''} (u, v) $
		}
	}
}

\BlankLine

\procedure{\delete{u, v}}{
	\increase{$u$, $v$, $\infty$} 
}

\BlankLine

\procedure{\increase{u, v, $w (u, v)$}}{
	\tcp{Increase weight of edge $ (u, v) $ to $ w(u, v) $}
	Insert $ u $ and $v$ into heap $ Q $ with keys $ \lev(u) $ and $\lev(v)$ respectively\;\label{line:insert u}
	Update key of $ v $ in heap $ N(u) $ to $ \lev(v) + w(u, v) $ and key of $ u $ in heap $ N(v) $ to $ \lev(u) + w(u, v) $\;\label{line:update N after increase}
	$ \updateLevels{} $\;
}

\BlankLine

\procedure{\insert{$u$, $v$, $ w (u, v)$}}{
	\tcp{Increase edge $ (u, v) $ of weight $ w(u, v) $}
	Insert $ v $ into heap $ N(u) $ with key $ \lev (v) + w (u, v) $ and $u$ into heap $N(v)$ with key $ \lev (u) + w_{H''}(u, v) $\; 
}

\BlankLine

\procedure{\updateLevels{}}{
	\While{heap $ Q $ is not empty}{
		Take node $ u $ with minimum key $ \lev (u) $ from heap $ Q $ and remove it from $ Q $\;
		$ \lev'(u) \gets \min_{v} (\lev (v) + w_{H''} (u, v)) $\;
		\tcp{$\min_{v} (\lev (v) + w_{H''} (u, v)) $ can be retrieved from the heap $ N(u) $. $\arg\min_{v} (\lev (v) + w_{H''} (u, v)) $ is $u$'s parent in the ES-tree. }
	
		\If{$ \lev'(u) > \lev (u) $}{\label{line:check_for_level_increase}
			$\lev(u)\gets \lev'(u)$\;\label{line:level_increase}
			\lIf{$ \lev' (u) > L $}{
				$ \lev (u) \gets \infty $
			}

		\ForEach{neighbor $ v $ of $ u $}{
				update key of $ u $ in heap $ N(v) $ to $ \lev(u) + w_{H''} (u, v) $\;
				insert $ v $ into heap $ Q $ with key $ \lev(v) $ if $Q$ does not already contain $ v $\;
			}
		}
	}
}
\end{algorithm}

For the analysis of the monotone ES-tree we will use the following terminology.
We say that an edge $ (u, v) $ is \emph{stretched} if $ \lev (u) > \lev (v) + w_{H''} (u, v) $.
We say that a node $ u $ is \emph{stretched} if it is incident to an edge $ (u, v) $ that is stretched.
Note that for a node $ u $ that is not stretched we have $ \lev (u) \leq \lev (v) + w_{H''} (u, v) $ for every edge $ (u, v) $ contained in $ H'' $.
In our proof we will use the following properties of the monotone ES-tree.

\begin{observation}[\cite{HenzingerKNSICOMP16}]\label{obs:simple_observations_monotone_ES}
The following holds for the monotone ES-tree:
\begin{enumerate}[label=(\arabic{*})]
\item \label{item: observation one} The level of a node never decreases.
\item \label{item: observation two} An edge can only become stretched when it is inserted.
\item \label{item: observation three} As long as a node is stretched, its level does not change. 
\item \label{item: observation four} For every tree edge $ (u, v) $ (where $ v $ is the parent of $ u $), $ \lev (u) \geq \lev (v) + w_{H''} (u, v) $.
\end{enumerate}
\end{observation}
Observe that property~\ref{item: observation four} above implies property~\ref{prop:no_underestimation_of_distance2}, i.e., that the returned distance estimate never underestimates the true distance.

A second prerequisite from~\cite{HenzingerKNSICOMP16} tells us when we may apply a variant of the triangle inequality to argue about the levels of nodes.

\begin{lemma}[\cite{HenzingerKNSICOMP16}]\label{lem:level_weight_inequality}
Let $ (u, v) $ be an edge of $ H'' $ such that $ \lev (v) + w_{H''} (u, v) \leq L $.
If $ (u, v) $ is not stretched and after the previous update in $ G $ the level of $ u $ was less than $ \infty $, then for the current level of~$ u $ we have $ \lev (u) \leq \lev (v) + w_{H''} (u, v) $.
\end{lemma}
Note that the second precondition simply captures the property of the monotone ES-tree that once the level of a node exceeds $ L $ it is set to $ \infty $ and will never be decreased anymore.
At the initialization (i.e., before the first update in $ H'' $), the first precondition is fulfilled automatically as no edge is stretched yet.

To count the additive error from rounding the edge weights, we define, for every node~$ u $ and every $ 0 \leq i \leq p-1 $, the function $ h (u, i) $ as follows:
\begin{equation*}
h (u, i) =
\begin{cases}
 0 & \text{if $ u = s $} \\
 (p+1) \left\lceil \frac{\max (\dist_G (u, s) - r_i, 0)}{\Delta} \right\rceil + p + 1 - i & \text{otherwise}
\end{cases}
\, .
\end{equation*}
The intuition is that $ h (u, i) $ bounds the number of hops from $ u $ to $ s $, i.e., the number of edges required to go from $ u $ to $ s $ while at the same time providing the desired approximation guarantee.
The approximation guarantee can now formally be stated as follows
\begin{lemma}[Approximation Guarantee]\label{lem:approximation_guarantee_priority}
For every $ 0 \leq i \leq p - 1 $ and every node $ u $ of priority~$ i $ with $ \dist_G (u, s) \leq D + \sum_{0 \leq i' \leq i-1} s_{i'} $ we have
\begin{equation*}
\distest (s, u) \leq (\alpha + \epsilon) \dist_G (u, s) + \gamma_i + h (u, i) \cdot \varphi \, .
\end{equation*}
\end{lemma}

Once we have proved this lemma, the desired bound on the approximation error (Property~\ref{prop:approximation_guarantee2}) follows easily because $ h (u, i) \cdot \varphi \leq \epsilon \dist_G (u, s) + 2 \epsilon \Delta $ (as we show below) and $ \gamma \leq \epsilon n^{1/p} \Delta $ by~\Cref{lem:bound_on_additive_error}, and thus
\begin{align*}
\distest (s, u) &\leq (\alpha + \epsilon) \dist_G (u, s) + \gamma_i + h (u, i) \cdot \varphi \\
 &\leq (\alpha + \epsilon) \dist_G (u, s) + \gamma_0 + h (u, i) \cdot \varphi \\
 &\leq (\alpha + \epsilon) \dist_G (u, s) + \gamma_0 + \epsilon \dist_G (u, s) + 2 \epsilon \Delta \\
 &= (\alpha + 2 \epsilon) \dist_G (u, s) + \gamma \\
 &\leq (\alpha + 2 \epsilon) \dist_G (u, s) + \epsilon n^{1/p} \Delta \, .
\end{align*}

\begin{lemma}\label{lem:upper_bound_on_total_rounding_error}
For every node $ u $ and every $ 0 \leq i \leq p - 1 $,
\begin{equation*}
h (u, i) \cdot \varphi \leq \epsilon \dist_G (u, s) + 2 \epsilon \Delta
\end{equation*}
\end{lemma}

\begin{proof}
If $ u = s $, then the claim is trivially true.
Otherwise we have
\begin{align*}
h (u, i) &= \left( (p+1) \left\lceil \frac{\max (\dist_G (u, s) - r_i, 0)}{\Delta} \right\rceil + p + 1 - i \right) \varphi \\
&\leq \left( (p+1) \left\lceil \frac{\dist_G (u, s)}{\Delta} \right\rceil + p + 1 \right) \varphi \\
&\leq \left( (p+1) \left( \frac{\dist_G (u, s)}{\Delta} + 1 \right) + p + 1 \right) \varphi \\
&= \left(\frac{(p+1) \dist_G (u, s)}{\Delta} + 2 (p+1) \right) \varphi \\
&= \left(\frac{(p+1) \dist_G (u, s)}{\Delta} + 2 (p+1) \right) \cdot \frac{\epsilon \Delta}{p+1} \\
&= \epsilon \dist_G (u, s) + 2 \epsilon \Delta \, . \qedhere
\end{align*}
\end{proof}

\begin{proof}[Proof of \Cref{lem:approximation_guarantee_priority}]
The proof is by double induction first on the number of updates in $ G $ and second on $ h (u, i) $.
Let $ 0 \leq i \leq p - 1 $ and let $ u $ be a node of priority~$ i $ such that $ \dist_G (u, s) \leq D + \sum_{0 \leq i' \leq i-1} s_{i'} $.
Remember that $ \distest (u, s) = \lev (u) \cdot \varphi $, where $ \lev (u) $ is the level of $ u $ in the monotone ES-tree of $ s $.
We know that after the previous deletion in $ G $ the distance estimate gave an approximation of the true distance in $ G $.
Since distances in $ G $ are non-decreasing it must have been the case that the level of $ u $ was less than $ \infty $ after the previous deletion in $ G $.

If $ u = s $, the claim is trivially true because $ \lev (s) = 0 $.
Assume that $ u \neq s $.
If $ u $ is stretched in the monotone ES-tree, then the level of $ u $ has not changed since the previous deletion in $ G $ and thus the claim is true by induction.
If $ u $ is not stretched, then $ \lev (u) \leq \lev (v) + w_{H''} (u, v) $ for every edge $ (u, v) $ in $ H'' $.
Define the nodes $ v $ and~$ x $ as follows.
If $ \dist_G (u, s) \leq r_i $, then $ v = s $.
If $ \dist_G (u, s) > r_i $, then consider a shortest path $ \pi $ from $ u $ to $ s $ in $ G $ and let $ v $ be the furthest node from $ u $ on $ \pi $ such that $ \dist_G (u, v) \leq r_i $ (which implies $ \dist_G (v, s) \geq \dist_G (u, s) - r_i $). 
Furthermore let~$ x $ be the neighbor of $ v $ on the shortest path $ \pi $ that is closer to $ s $ than $ v $ is.
Note that $ \dist_G (u, x) \geq r_i $ (and thus $ \dist_G (x, s) \leq \dist_G (u, s) - r_i $) and in particular $ G $ contains the edge $ (v, x) $.
The edge $ (v, x) $ is also contained in~$ H $ (and thus in~$ H' $ and~$ H'' $) by the following argument:
For $ \dist_G (u, s) \leq D + \sum_{0 \leq i' \leq i-1} s_{i'} $ to hold it has to be the case that $ w_G (v, x) \leq D + \sum_{0 \leq i' \leq i-1} s_{i'} $.
Note that $ \sum_{0 \leq i' \leq i-1} s_{i'} \leq \sum_{0 \leq i' \leq i-1} w_{i'} \leq r_{p-1} \leq n^{1/p} \Delta $ by \Cref{lem:bound_on_radius}.
Thus, $ w_G (v, x) \leq D + n^{1/p} \Delta $, which by the definition of $ H $ means that the edge $ (v, x) $ is contained in $ H $.

Note that $ s (\dist_G (u, v), p - 1 - i) \leq s (r_i, p - 1 - i) $ since the function $ s (\cdot, \cdot) $ is non-decreasing in the first argument.
By \Cref{lem:auxiliary_inequality_distance_to_next_node} we have $ s (r_i, p - 1 - i) \leq r_{p - 1} $ and by \Cref{lem:bound_on_radius} we have $ r_{p-1} \leq n^{1/p} \Delta \leq \hat{D} $.
It follows that $ s (\dist_G (u, v), p - 1 - i) \leq \hat{D} $.
Thus, by Property~\ref{prop:finding_node_of_higher_priority2} we know that either $ v \in B (u) $ or there is a node $ v' $ of priority $ j' > i $ such that $ u \in B (v) $ and $ \dist_G (u, v') \leq s (\dist_G (u, v), j' - i) $.
Note that in the first case the set of edges~$ F $ contains the edge $ (u, v) $ and in the second case it contains the edge $ (u, v') $.

\noindent \underline{Case 1: $ v \in B (u) $}

If $ v \in B (u) $, then $ F $ contains an edge $ (u, v) $ such that
\begin{equation}
w_F (u, v) = \hat{\distest} (u, v) \leq \alpha \dist_G (u, v) + \beta \label{eq:approximation_of_weight_in_F}
\end{equation}
Since $ \dist_G (u, v) \leq r_i $ we have $ w_F (u, v) \leq \alpha r_i + \beta \leq \alpha r_{p-1} + \beta \leq n^{1/p} \Delta $, where the last inequality holds by \Cref{lem:bound_on_radius}.
Thus, $ (u, v) $ is contained in~$ H $ and thus also in $ H' $ and~$ H'' $.

If $ \dist_G (u, s) \leq r_i $, then we have $ v = s $.
First observe that by the definition of $ H'' $ we have $ w_{H''} (u, s) = w_{H'} (u, s) / \varphi $.
Furthermore the rounding of the edge weights in $ H' $ guarantees that $ w_{H'} (u, s) \leq w_H (u, s) + \varphi $.
We therefore get
\begin{align*}
w_{H''} (u, s) &\leq \frac{w_F (u, s) + \varphi}{\varphi} \\
 &\leq \frac{\alpha \dist_G (u, s) + \beta + \varphi}{\varphi} \\
 &\leq \frac{\alpha \left( D + \sum_{0 \leq i' \leq i-1} s_{i'} \right) + \beta + \varphi}{\varphi} \\
 &\leq \frac{ \alpha D + (\alpha + 1 + \epsilon) \sum_{0 \leq i' \leq p-2} w_{i'} + \beta + \varphi}{\varphi} \\
 &= \frac{\alpha D + \gamma_0 + \varphi}{\varphi} \\
 &= \frac{\alpha D + \gamma_0 + \frac{\epsilon \Delta}{p+1}}{\varphi} \\
 &\leq \frac{\alpha D + \gamma_0 + 2 \epsilon \Delta}{\varphi} \\
 &= \frac{\alpha D + \gamma}{\varphi}
 \leq \frac{\alpha D + \epsilon n^{1/p} \Delta}{\varphi} \leq \frac{(\alpha + 2 \epsilon) D}{\varphi} + (p+1) n^{1/p} = L \, .
\end{align*}
Here we have used the inequality $ \gamma \leq \epsilon n^{1/p} \Delta $ from \Cref{lem:bound_on_additive_error}.
Since the maximum level in the monotone ES-tree is $ L $ and $ u $ is not stretched, it follows from \Cref{lem:level_weight_inequality} that $ \lev (u) \leq \lev (s) + w_{H''} (u, s) = w_{H''} (u, s) $.
Together with the observations $ h (u, i) \geq 1 $ (since $ u \neq s $) and $ \beta \leq \gamma_0 $ we therefore get
\begin{multline*}
\distest (s, u) = \lev (u) \cdot \varphi \leq w_{H''} (u, s) \cdot \varphi \leq \alpha \dist_G (u, s) + \beta + \varphi \\
\leq \alpha \dist_G (u, s) + \beta + h (u, i) \cdot \varphi \leq (\alpha + \epsilon) \dist_G (u, s) + \gamma_0 + h (u, i) \cdot \varphi \, .
\end{multline*}

Consider now the case $ \dist_G (u, s) > r_i $.
Let $ j $ denote the priority of $ x $.
We first prove the following inequality, which will allow us among other things to use the induction hypothesis on $ x $.

\begin{claim}\label{claim:hops_inequality_1}
If $ \dist_G (u, s) > r_i $, then $ h (x, j) + 2 \leq h (u, i) $.
\end{claim}

\begin{proof}
Remember that $ i \leq p-1 $.
The assumption $ \dist_G (u, s) > r_i $ implies that $ \dist_G (x, s) \leq \dist_G (u, s) - r_i $.
If $ \dist_G (x, s) < r_j $, we have
\begin{align*}
h (x, j) + 2 \leq p + 1 - j + 2 \leq p + 1 + 2 &\leq p + 1 + p + 1 - i \\
 &\leq (p+1) \left\lceil \frac{\dist_G (u, s) - r_i}{\Delta} \right\rceil + p + 1 - i = h (u, i) \, .
\end{align*}
Here we use the inequality $ \lceil (\dist_G (u, s) - r_j) / \Delta \rceil \geq 1 $ which follows from the assumption $ \dist_G (u, s) > r_i $.

If $ \dist_G (x, s) \geq r_j $, then, using $ r_j \geq r_0 \geq \Delta $, we get
\begin{align*}
h (x, j) + 2 &= (p+1) \left\lceil \frac{\dist_G (x, s) - r_j}{\Delta} \right\rceil + p + 1 - j + 2 \\
 &\leq (p+1) \left\lceil \frac{\dist_G (x, s) - \Delta}{\Delta} \right\rceil + p + 1 + 2 \\
 &= (p+1) \left\lceil \frac{\dist_G (x, s)}{\Delta} - 1 \right\rceil + p + 1 + 2 \\
 &= (p+1) \left( \left\lceil \frac{\dist_G (x, s)}{\Delta} \right\rceil - 1 \right) + p + 1 + 2 \\
 &= (p+1) \left\lceil \frac{\dist_G (x, s)}{\Delta} \right\rceil + 2 \\
 &\leq (p+1) \left\lceil \frac{\dist_G (x, s)}{\Delta} \right\rceil + p + 1 - i \\
 &\leq (p+1) \left\lceil \frac{\dist_G (u, s) - r_i}{\Delta} \right\rceil + p + 1 - i \\
 &\leq (p+1) \left\lceil \frac{\max (\dist_G (u, s) - r_i, 0)}{\Delta} \right\rceil + p + 1 - i = h (u, i) \, .
\end{align*}
Here the last inequality follows from the trivial observation $ \dist_G (u, s) - r_i \leq \max (\dist_G (u, s) - r_i, 0) $.
\end{proof}

Having proved this claim, we go on with the proof of the lemma.
We will now show that
\begin{equation}
\lev (x) + w_{H''} (v, x) + w_{H''} (u, v) \leq \frac{(\alpha + \epsilon) \dist_G (u, s) + \gamma_i + h (u, i) \cdot \varphi}{\varphi} \label{eq:bound_on_tentative_level_case1}
\end{equation}
as follows.
If $ \dist_G (u, s) > r_i $, then we have $ \dist_G (u, x) \geq r_i $ by the choice of $ x $.
Remember that the edge $ (v, x) $ lies on a shortest path from $ u $ to $ s $ in $ G $.
It is therefore contained in~$ G $ since before the first deletion and thus will never be stretched.
We also may apply the induction hypothesis on $ x $ since
\begin{equation*}
\dist_G (x, s) = \dist_G (u, s) - \dist_G (u, x) \leq \dist_G (u, s) - r_i \leq D + \sum_{0 \leq i' \leq i-1} s_{i'} - r_i \leq D
\end{equation*}
due to $ \sum_{0 \leq i' \leq i-1} s_{i'} \leq r_i $ by the definition of $ r_i $.
Therefore we get
\begin{align*}
(\lev (x) &+ w_{H''} (v, x) + w_{H''} (u, v)) \cdot \varphi \\
 &\leq \distest (s, x) + w_{H''} (v, x) \cdot \varphi + w_{H''} (u, v) \cdot \varphi && \text{\small (definition of $ \distest (s, x) $)} \\
 &= \distest (s, x) + w_{H'} (v, x) + w_{H'} (u, v) && \text{\small (definition of $ H' $)} \\
 &\leq \distest (s, x) + w_H (v, x) + \varphi + w_H (u, v) + \varphi && \text{\small (property of $ w_{H'} $)} \\
 &\leq \distest (s, x) + w_G (v, x) + \varphi + w_F (u, v) + \varphi && \text{\small ($ (v, x) \in E $ and $ (u, v) \in F $)} \\
 &\leq  (\alpha + \epsilon) \dist_G (x, s) + \gamma_j + h (x, j) \cdot \varphi + w_G (v, x) + \varphi + w_F (u, v) + \varphi && \text{\small (induction hypothesis)} \\
 &= (\alpha + \epsilon) \dist_G (x, s) + \gamma_j + w_F (u, v) + w_G (v, x) + (h (x, j) + 2) \cdot \varphi && \text{\small (rearranging terms)} \\
 &\leq (\alpha + \epsilon) \dist_G (x, s) + \gamma_j + w_F (u, v) + w_G (v, x) + h (u, i) \cdot \varphi && \text{\small (Claim~\ref{claim:hops_inequality_1})} \\
 &\leq (\alpha + \epsilon) \dist_G (x, s) + \gamma_0 + w_F (u, v) + w _G(v, x) + h (u, i) \cdot \varphi && \text{\small ($ \gamma_j \leq \gamma_0 $)} \\
 &\leq (\alpha + \epsilon) \dist_G (x, s) + \gamma_0 + \alpha \dist_G (u, v) + \beta + w_G (v, x) + h (u, i) \cdot \varphi && \text{\small (by Inequality~\eqref{eq:approximation_of_weight_in_F})} \\
 &= (\alpha + \epsilon) \dist_G (x, s) + \gamma_0 + \alpha \dist_G (u, v) + \beta + \dist_G (v, x) + h (u, i) \cdot \varphi && \text{\small ($ (v, x) $ on shortest path)} \\
 &\leq (\alpha + \epsilon) \dist_G (x, s) + \gamma_0 + \alpha \dist_G (u, v) + \beta + \alpha \dist_G (v, x) + h (u, i) \cdot \varphi && \text{\small ($ \alpha \geq 1 $)} \\
 &= (\alpha + \epsilon) \dist_G (x, s) + \alpha (\dist_G (u, v) + \dist_G (v, x)) + \beta + \gamma_0 + h (u, i) \cdot \varphi && \text{\small (rearranging terms)} \\
 &= (\alpha + \epsilon) \dist_G (x, s) + \alpha \dist_G (u, x) + \beta + \gamma_0 + h (u, i) \cdot \varphi && \text{\small ($ v $ on shortest path)} \\
 &= (\alpha + \epsilon) \dist_G (x, s) + \alpha \dist_G (u, x) + \beta + \gamma_0 - \gamma_i + \gamma_i + h (u, i) \cdot \varphi && \text{\small (zero addition)} \\
 &= (\alpha + \epsilon) \dist_G (x, s) + \alpha \dist_G (u, x) + \epsilon r_i + \gamma_i + h (u, i) \cdot \varphi && \text{\small (by \Cref{lem:alternative_formulation_of_radius})} \\
 &\leq (\alpha + \epsilon) \dist_G (x, s) + \alpha \dist_G (u, x) + \epsilon \dist_G (u, x) + \gamma_i + h (u, i) \cdot \varphi && \text{\small ($ \dist_G (u, x) \geq r_i $)} \\
 &= (\alpha + \epsilon) (\dist_G (u, x) + \dist_G (x, s)) + \gamma_i + h (u, i) \cdot \varphi && \text{\small (rearranging terms)} \\
 &= (\alpha + \epsilon) \dist_G (u, s) + \gamma_i + h (u, i) \cdot \varphi && \text{\small ($ x $ on shortest path)} \, .
\end{align*}

By \Cref{lem:upper_bound_on_total_rounding_error} we have $ h (u, i) \cdot \varphi \leq \epsilon \dist_G (u, s) + 2 \epsilon \Delta $ and thus Inequality~\eqref{eq:bound_on_tentative_level_case1} implies that
\begin{align*}
\lev (x) + w_{H''} (v, x) + w_{H''} (u, v) &\leq \frac{(\alpha + 2 \epsilon) \dist_G (u, s) + \gamma_i + 2 \epsilon \Delta}{\varphi} \\
 &\leq \frac{(\alpha + 2 \epsilon) \left( D + \sum_{0 \leq i' \leq i-1} s_{i'} \right) + \gamma_i + 2 \epsilon \Delta}{\varphi} \\
 &\leq \frac{(\alpha + 2 \epsilon) D + (\alpha + 1 + \epsilon) \left( \sum_{0 \leq i' \leq i-1} w_{i'} \right) + \gamma_i + 2 \epsilon \Delta}{\varphi} \\
 &\leq \frac{(\alpha + 2 \epsilon) D + \gamma_0 + 2 \epsilon \Delta}{\varphi} \\
 &= \frac{(\alpha + 2 \epsilon) D + \gamma}{\varphi} \\
 &\leq \frac{(\alpha + 2 \epsilon) D + \epsilon n^{1/p} \Delta}{\varphi} \\
 &= \frac{(\alpha + 2 \epsilon) D}{\varphi} + (p+1) n^{1/p} = L
\end{align*}
where the last inequality follows from \Cref{lem:bound_on_additive_error}.
As the maximum level in the monotone ES-tree is $ L $ and the edge $ (v, x) $ is not stretched, it follows from \Cref{lem:level_weight_inequality} that $ \lev (v) \leq \lev (x) + w_{H''} (v, x) $ and since $ u $ is not stretched, we have
\begin{equation*}
\lev (u) \leq \lev (v) + w_{H''} (u, v) \leq \lev (x) + w_{H''} (v, x) + w_{H''} (u, v) \, .
\end{equation*}
and thus
\begin{equation*}
\distest (s, u) = \lev (u) \cdot \varphi \leq (\lev (x) + w_{H''} (v, x) + w_{H''} (u, v)) \cdot \varphi \leq (\alpha + \epsilon) \dist_G (u, s) + \gamma_i + h (u, i) \cdot \varphi
\end{equation*}

\noindent \underline{Case 2: $ v \notin B (u) $}

By Property~\ref{prop:finding_node_of_higher_priority2} we know that there is some node $ v' $ of priority $ j' > i $ such that $ u \in B (v') $ and $ \dist_G (u, v') \leq s (\dist_G (u, v), j'-i) $.
By \Cref{lem:auxiliary_inequality_distance_to_next_node} we therefore have
\begin{equation*}
\dist_G (u, v') \leq s (r_i, j'-i) \leq s (r_{j'-1}, 1) = s_{j'-1} \, .
\end{equation*}
From the definition of $ F $ and Property~\ref{prop:approximation_in_balls2} it now follows that $ F $ contains the edge $ (u, v') $ of weight
\begin{equation}
\dist_G (u, v') \leq w_F (u, v') = \hat{\distest} (u, v') \leq \alpha \dist_G (u, v') + \beta \leq \alpha s_{j'-1} + \beta = w_{j'-1} \label{eq:upper_and_lower_bound_on_hop_set_edge}
\end{equation}
Since $ j' \leq p-1 $ we have $ w_{j'-1} \leq w_{p-2} \leq r_{p-1} $.
As $ r_{p-1} \leq n^{1/p} \Delta $, by \Cref{lem:bound_on_radius}, we conclude that the edge $ (u, v') $ is contained $ H $ and thus also in $ H' $ and $ H'' $.

We first prove the following inequality, which will allow us among other things to apply the induction hypothesis on $ v' $.

\begin{claim}\label{claim:hops_inequality_2}
$ h (v', j') + 1 \leq h (u, i) $
\end{claim}

\begin{proof}
Remember that $ j' \geq i + 1 $.
If $ \dist_G (v', s) < r_{j'} $, we get
\begin{equation*}
h (v', j') + 1 \leq p + 1 - j' + 1 \leq p + 1 - i \leq h (u, i) \, .
\end{equation*}

If $ \dist_G (v', s) \geq r_{j'} $, then we use the inequality $ r_{j'} \geq r_i + s_{j'-1} $ (which easily follows from the definition of $ r_{j'} $ and the fact that $ \alpha \geq 1 $) and get
\begin{align*}
h (v', j') + 1 &= (p+1) \left\lceil \frac{\dist_G (v', s) - r_{j'}}{\Delta} \right\rceil + p + 1 - j' + 1 \\
 &\leq (p+1) \left\lceil \frac{\dist_G (v', s) - r_{j'}}{\Delta} \right\rceil + p + 1 - i - 1 + 1 \\
 &\leq (p+1) \left\lceil \frac{\dist_G (v', u) + \dist_G (u, s) - r_{j'}}{\Delta} \right\rceil + p + 1 - i \\
 &\leq (p+1) \left\lceil \frac{s_{j'-1} + \dist_G (u, s) - r_{j'}}{\Delta} \right\rceil + p + 1 - i \\
 &\leq (p+1) \left\lceil \frac{\dist_G (u, s) - r_i}{\Delta} \right\rceil + p - i \\
 &\leq (p+1) \left\lceil \frac{\max(\dist_G (u, s) - r_i, 0)}{\Delta} \right\rceil + p + 1 - i = h (u, i) \, . \qedhere
\end{align*}
\end{proof}

Having proved this claim, we go on with the proof of the lemma.
Note that we may apply the induction hypothesis on $ v' $ because by the triangle inequality we have
\begin{align*}
\dist_G (v', s) \leq \dist_G (u, s) + \dist_G (v', u) &\leq D + \sum_{0 \leq i' \leq i-1} s_{i'} + \dist_G (v', u) \\
 &\leq D + \sum_{0 \leq i' \leq i-1} s_{i'} + s_{j'-1} \leq D + \sum_{0 \leq i' \leq j'-1} s_{i'} \, .
\end{align*}

We will now show that
\begin{equation}
\lev (v') + w_{H''} (u, v') \leq \frac{(\alpha + \epsilon) \dist_G (u, s) + \gamma_i + h (u, i) \cdot \varphi}{\varphi} \label{eq:bound_on_tentative_level_case2}
\end{equation}
as follows:
\begin{align*}
(\lev (v') &+ w_{H''} (u, v')) \cdot \varphi && \text{\small ($ u $ not stretched)} \\
 &= \distest (v', s) + w_{H''} (u, v') \cdot \varphi && \text{\small (definition of $ \distest (v', s) $)} \\
 &= \distest (v', s) + w_{H'} (u, v') && \text{\small (definition of $ H'' $)} \\
 &\leq \distest (v', s) + w_{H} (u, v') + \varphi && \text{\small (property of $ w_{H'} (u, v') $)} \\
 &\leq \distest (v', s) + w_{F} (u, v') + \varphi && \text{\small (definition of $ H $)} \\
 &\leq (\alpha + \epsilon) \dist_G (v', s) + \gamma_{j'} + h (v', j') \cdot \varphi + w_F (u, v') + \varphi && \text{\small (induction hypothesis)} \\
 &= (\alpha + \epsilon) \dist_G (v', s) + \gamma_{j'} + w_F (u, v') + (h (v', j') + 1) \cdot \varphi && \text{\small (rearranging terms)} \\
 &\leq (\alpha + \epsilon) \dist_G (v', s) + \gamma_{j'} + w_F (u, v') + h (u, i) \cdot \varphi && \text{\small (Claim~\ref{claim:hops_inequality_2})} \\
 &\leq (\alpha + \epsilon) (\dist_G (v', u) + \dist_G (u, s) ) + \gamma_{j'} + w_F (u, v') + h (u, i) \cdot \varphi && \text{\small (triangle inequality)} \\
 &\leq (\alpha + \epsilon) (w_F (u, v') + \dist_G (u, s) ) + \gamma_{j'} + w_F (u, v') + h (u, i) \cdot \varphi && \text{\small (by Inequality~\eqref{eq:upper_and_lower_bound_on_hop_set_edge})} \\
 &= (\alpha + \epsilon) \dist_G (u, s) + \gamma_{j'} + (\alpha + \epsilon + 1) w_F (u, v') + h (u, i) \cdot \varphi && \text{\small (rearranging terms)} \\
 &\leq (\alpha + \epsilon) \dist_G (u, s) + \gamma_{j'} + (\alpha + \epsilon + 1) w_{j'-1} + h (u, i) \cdot \varphi && \text{\small (by Inequality~\eqref{eq:upper_and_lower_bound_on_hop_set_edge})} \\
 &= (\alpha + \epsilon) \dist_G (u, s) + \gamma_{j'-1} + h (u, i) \cdot \varphi && \text{\small (definition of $ \gamma_{j'-1} $)} \\
 &\leq (\alpha + \epsilon) \dist_G (u, s) + \gamma_i + h (u, i) \cdot \varphi && \text{\small ($ \gamma_i \geq \gamma_{j'-1} $ as $ j' \geq i+1 $)} \, .
\end{align*}

By \Cref{lem:upper_bound_on_total_rounding_error} we have $ h (u, i) \cdot \varphi \leq \epsilon \dist_G (u, s) + 2 \epsilon \Delta $ and thus Inequality~\eqref{eq:bound_on_tentative_level_case2} implies that
\begin{equation*}
\lev (v') + w_{H''} (u, v') \leq \frac{(\alpha + 2 \epsilon) \dist_G (u, s) + \gamma_i + 2 \epsilon \Delta}{\varphi}
 \leq \frac{(\alpha + 2 \epsilon) D}{\varphi} + (p+1) n^{1/p} = L \, .
\end{equation*}
As the maximum level in the monotone ES-tree is $ L $ and $ u $ is not stretched, it follows from \Cref{lem:level_weight_inequality} that $ \lev (u) \leq \lev (v') + w_{H''} (u, v') $ and thus
\begin{align*}
\distest (s, u) &= \lev (u) \cdot \varphi \leq (\lev (v') + w_{H''} (u, v')) \cdot \varphi \leq (\alpha + \epsilon) \dist_G (u, s) + \gamma_i + h (u, i) \cdot \varphi \, . \qedhere
\end{align*}
\end{proof}

%% file: putting_together.tex
\section{Putting Everything Together}\label{sec:putting_together}

In the following we combine the results of \Cref{sec:from_sssp_to_balls} and \Cref{sec:from_balls_to_sssp} to obtain decremental algorithms for approximate SSSP and approximate APSP.

\subsection{Approximate SSSP}

We first show how to obtain an algorithm for approximate SSSP.
First, we obtain an algorithm that provides approximate distance for all nodes that are at distance of at most $ R $ from the source, where $ R $ is some range parameter.
We use a hierarchical approach to obtain this algorithm:
Given an algorithm for maintaining approximate shortest paths, we obtain an algorithm for maintaining approximate balls, which in turn gives us an algorithm for maintaining approximate shortest paths for a larger range of distances than the initial algorithm.
This scheme is repeated several times and can be ``started'' with the (exact) ES-tree.

\begin{lemma}\label{lem:restricted_approximate_shortest_paths}
For every $ R \geq n $ and every $ 0 < \epsilon \leq 1 $, there is a decremental approximate SSSP algorithm that, given a fixed source node $ s $, maintains, for every node $ v $, a distance estimate $ \distest (s, v) $ such that $ \distest (s, v) \geq \dist_G (s, v) $ and if $ \dist_G (s, v) \leq R $, then $ \distest (s, v) \leq (1 + \epsilon) \dist_G (s, v) $.
It has a total update time of $ O (m^{1 + O ((\log{\log{R}})/q)} R^{2/q} + n) $ in expectation, where
\begin{equation*}
q = \left\lfloor \sqrt{ \left\lfloor \frac{\sqrt{\log{n}}}{\sqrt{\log{\left( \frac{12 \cdot 4^3 \log{n}}{\epsilon} \right)}}} \right\rfloor } \right\rfloor
\end{equation*}
and, after every update in $ G $, returns each node $ v $ such that $ \distest (s, v) $ has changed together with the new value of $ \distest (s, v) $.
\end{lemma}

\begin{proof}
In the proof we will use the following values.
We set $ a = 4 $,
\begin{equation*}
p = \left\lfloor \frac{\sqrt{\log{n}}}{\sqrt{\log{\left( \frac{12 a^3 \log{n}}{\epsilon} \right)}}} \right\rfloor
\end{equation*}
and $ q = \lfloor \sqrt{p} \rfloor $.
Furthermore we set $ \epsilon' = \epsilon / (2 (q-2)) $ and for every $ 0 \leq k \leq q-2 $ we set $ \alpha_k = 1 + 3 k \epsilon' \leq 1 + \epsilon $, $ a_k = 2 \alpha_k \leq a $, $ \Delta_k = R^{k/q} $, and $ D_k = R^{(k+2)/q} $.

The heart of our proof is the following claim which gives us decremental approximate SSSP algorithms for larger and larger depths, until finally the full range $ R $ is covered.

\begin{claim}
For every $ 0 \leq k \leq q-2 $, there is a decremental approximate SSSP algorithm $ \textsc{ApproxSSSP}_k $ with the following properties:
\begin{enumerate}[label=\textbf{A\arabic*}]
\item $ \distest (s, v) \geq \dist_G (s, v) $
\item If $ \dist_G (s, v) \leq D_k $, then $ \distest (s, v) \leq \alpha_k \dist_G (s, v) $.
\item The expected total update time of $ \textsc{ApproxSSSP}_k $ is
\begin{equation*}
T_k (m, n) = O (p \log{n} \log{R})^k \cdot O (m^{1 + k/p} R^{2/q} / \epsilon') + O (n) \, .
\end{equation*}
\item After every update in $ G $, $ \textsc{ApproxSSSP}_k $ returns each node $ v $ such that $ \distest (s, v) $ has changed together with the new value of $ \distest (s, v) $.
\end{enumerate}
\end{claim}

\begin{proof}
We prove the claim by induction on $ k $.
In the base case $ k = 0 $ we use the (exact) ES-tree (see \Cref{lem:ES-tree}), which for distances up to $ D \leq D_0 $ has a total update time of $ O (m D_0 + n) = O (m R^{2/q} + n) $ and thus has all claimed properties

We now consider the induction step.
We apply \Cref{pro:from_sssp_to_balls} to obtain a decremental algorithm $ \textsc{ApproxBalls}_k $ (with parameters $ \hat{\p} = p $ and~$ \hat{\epsilon} = 1 $) that maintains for every node $ u \in V $ a set of nodes $ B_k (u) $ and a distance estimate $ \hat{\distest}_k (u, v) $ for every node $ v \in B_k (u) $ such that:
\begin{enumerate}[label=\textbf{B\arabic*}]
\item For every node $ u $ and every node $ v \in B_k (u) $ we have $ \dist_G (u, v) \leq \hat{\distest}_k (u, v) \leq  \alpha_{k-1} \dist_G (u, v) $.
\item For all $ x \geq 0 $, set $ s_k (x, 0) = x $, and for all $ x \geq 0 $ and $ l \geq 1 $, set $ s_k (x, l) = a_{k-1} (a_{k-1} + 1)^{l-1} x $. Then for every $ 0 \leq i \leq p - 1 $, every node~$ u $ of priority~$ i $, and every node~$ v $ such that $ s_k (\dist_G (u, v), p - 1 - i) \leq D_k $, either (1) $ v \in B_k (u) $ or (2) there is some node $ v' $ of priority $ j > i $ such that $ u \in B_k (v') $ and $ \dist_G (u, v') \leq s_k (\dist_G (u, v), j - i) $.
\item In expectation, $ \sum_{u \in V} \mathcal{B}_k (u) = O (p m^{1+1/p} \log{D_k}) $, where $ \mathcal{B}_k (u) $ denotes the number of nodes ever contained in $ B_k (u) $.
\item The total update time of $ \textsc{ApproxBalls}_k $ is
\begin{equation*}
t_k (m, n) = O \left( \left(p m^{1 + 1/p} + \sum_{0 \leq i \leq p - 1} \frac{m}{m^{i/p}} \cdot T_{k-1} (m_i, n_i) \right) \log{n} \log{D_k} + p T_{k-1} (m, n) \right)
\end{equation*}
in expectation, where, for each $ 0 \leq i \leq p - 1 $, $ m_i = O (m^{(i+1)/p}) $ and $ n_i = O (m^{(i+1)/p}) $.
\end{enumerate}

Note that $ D_k \leq R $ and thus $ \log{D_k} \leq \log{R} $ and remember that by the induction hypothesis we have
\begin{equation*}
T_k (m, n) = O (p \log{n} \log{R})^{k-1} \cdot O (m^{1 + (k-1)/p} R^{2/q} / \epsilon') + O (n) \, .
\end{equation*}
To analyze $ \tfrac{m}{m^{i/p}} \cdot T_{k-1} (m_i, n_i) $ for every $ 0 \leq i \leq p - 1 $, observe that $ \tfrac{m}{m^{i/p}} \cdot (m^{(i+1)/p})^{1+(k-1)/p} \leq m^{1 + k/p} $
because
\begin{align*}
1 - i/p + ((i+1)/p) \cdot (1 + (k-1)/p) &= 1 + 1/p + ((i+1)/p) ((k-1)/p) \\
 &\leq 1 + 1/p + (k-1)/p \\
 &= 1 + k/p \, .
\end{align*}
It now follows that
\begin{equation*}
t_k (m, n) = O (p \log{n} \log{R})^k \cdot O (m^{1 + k/p} R^{2/q} / \epsilon') + O (n) \, .
\end{equation*}

We now want to argue that we may apply \Cref{pro:from_balls_to_sssp} to obtain an approximate decremental SSSP algorithm $ \textsc{ApproxSSSP}_k' $ (with parameters $ p $, $ \Delta_k $ $ D_k $, and $ \epsilon' $).
We first show that
\begin{equation*}
p \leq \frac{\sqrt{\log{n}}}{\sqrt{\log{\left( \frac{4 a^3}{\epsilon'} \right)}}} \, ,
\end{equation*}
First note that $ q \leq \log{n} $ and thus $ \epsilon' = \epsilon ( 2 (q-2) \geq \epsilon / (2 q) \geq \epsilon / (2 \log{n}) $.
It follows that
\begin{equation*}
\frac{\sqrt{\log{n}}}{\sqrt{\log{\left( \frac{4 a^3}{\epsilon'} \right)}}} \geq \frac{\sqrt{\log{n}}}{\sqrt{\log{\left( \frac{12 a^3 \log{n}}{\epsilon} \right)}}} \geq \left\lfloor \frac{\sqrt{\log{n}}}{\sqrt{\log{\left( \frac{12 \cdot 4^3 \log{n}}{\epsilon} \right)}}} \right\rfloor = p \, .
\end{equation*}
Note also that for all $ x \geq 0 $ we have $ s_k (x, 1) = a_{k-1} x $ and for all $ x \geq 0 $ and $ l \geq 1 $ we have
\begin{equation*}
s_k (x, l + 1) = (a_{k-1} + 1) s_k (x, l) \leq 2 a_{k-1} s_k (x, l) \leq (\alpha_{k-1} + 1 + \epsilon') \alpha_{k-1} a_{k-1} s_k (x, l) \, .
\end{equation*}
We therefore may apply \Cref{pro:from_balls_to_sssp} to obtain an approximate decremental SSSP algorithm $ \textsc{ApproxSSSP}_k' $ (with parameters $ p $, $ \Delta_k $, $ D_k $, and $ \epsilon' $) that maintains, for every node $ v \in V $, a distance estimate $ \distest' (s, v) $ such that:
\begin{enumerate}[label=\textbf{A\arabic*'}]
\item $ \distest' (s, v) \geq \dist_G (s, v) $
\item If $ \dist_G (s, v) \leq D_k $, then $ \distest' (s, v) \leq (\alpha_k + 2 \epsilon') \dist_G (s, v) + \epsilon' n^{1/p} \Delta_k $
\item The total update time of $ \textsc{ApproxSSSP}_k' $ is
\begin{equation*}
T_k' (m, n) = t_k (m, n) + O \left(p \left( \alpha_k D_k / \Delta_k + n^{1/p} \right) \left(m + \sum_{u \in V} \mathcal{B}_k (u) \right) / \epsilon' + n \right) \, .
\end{equation*}
\item After every update in $ G $, $ \textsc{ApproxSSSP}_k' $ returns each node $ v $ such that $ \distest (s, v) $ has changed together with the new value of $ \distest (s, v) $.
\end{enumerate}

Note that $ \alpha_k \leq 1 + \epsilon \leq 2 $ and $ D_k / \Delta_k = R^{2/q} $.
Since $ q \leq p $ and $ R \geq n $ we have $ n^{1/p} \leq R^{2/q} $.
We also have $ \sum_{u \in V} \mathcal{B}_k (u) = O (p m^{1 + 1/p} \log{R}) $ in expectation.
Therefore the expected total update time of $ \textsc{ApproxSSSP}_k' $ is
\begin{align*}
T_k' (m, n) &= O (p \log{n} \log{R})^k \cdot O (m^{1 + k/p} R^{2/q} / \epsilon') + O (p^2 m^{1+1/p} R^{2/q} \log{R} / \epsilon' + n)
\end{align*}
and since $ p \leq \log{n} $ it follows that
\begin{equation*}
T_k' (m, n) = O (p \log{n} \log{R})^k \cdot O (m^{1 + k/p} R^{2/q} / \epsilon') + O (n) \, .
\end{equation*}

Let $ \textsc{ApproxSSSP}_k $ denote the algorithm that internally runs both $ \textsc{ApproxSSSP}_k' $ and $ \textsc{ApproxSSSP}_{k-1} $ and additionally maintains, for every node $ v $, the value $ \distest_k (s, v) = \min (\distest_k' (s, v), \distest_{k-1} (s, v)) $.
Since both $ \textsc{ApproxSSSP}_k' $ and $ \textsc{ApproxSSSP}_{k-1} $ return, after each update in $ G $, every node $ v $ for which $ \distest (s, v) $ has changed, and the minimum can be computed in constant time, $ \textsc{ApproxSSSP}_k $ has the same asymptotic total update time as $ \textsc{ApproxSSSP}_k' $.
It remains to show that $ \distest_k (s, v) $ fulfills the desired approximation guarantee for every node $ v $.
Since both $ \distest_k' (s, v) \geq \dist_G (s, v) $ and $ \distest_{k-1} (s, v) \geq \dist_G (s, v) $ also $ \distest_k (s, v) \geq \dist_G (s, v) $.
Furthermore, we know that if $ \dist_G (s, v) \leq D_k $, then $ \distest_k' (s, v) \leq \epsilon n^{1/p} \Delta_k $.
Let $ v $ be a node such that $ \dist_G (s, v) \leq D_k $.
If $ \dist_G (s, v) \leq D_{k-1} $, then $ \distest_k (s, v) \leq \distest_{k-1} (s, v) \leq \alpha_{k-1} \dist_G (s, v) \leq \alpha_k \dist_G (u, v) $.
If $ \dist_G (s, v) \geq D_{k-1} $, then
\begin{align*}
\distest_k (s, v) \leq \distest_k' (s, v) &\leq (\alpha_{k-1} + 2 \epsilon') \dist_G (s, v) + \epsilon' n^{1/p} \Delta_k \\
&\leq (\alpha_{k-1} + 2 \epsilon') \dist_G (s, v) + \epsilon' D_{k-1} \leq (\alpha_{k-1} + 3 \epsilon') \dist_G (s, v) \\
&= \alpha_k \dist_G (s, v) \, .
\end{align*}
This finishes the proof of the claim.
\end{proof}

The lemma now follows from the claim by observing that $ \textsc{ApproxSSSP}_{q-2} $ is the desired decremental approximate SSSP algorithm.
The correctness simply follows from the choice $ D_{q-2} = R $.
The expected total update time is
\begin{equation*}
T_{q-2} (m, n) = O (p \log{n} \log{R})^{q-2} \cdot O (m^{1 + (q-2)/p} R^{2/q} / \epsilon') + O (n) \, .
\end{equation*}
Remember that $ q = \lfloor \sqrt{p} \rfloor $ and thus $ (q-2)/p \leq q/p \leq 1 / \sqrt{p} \leq 1 / q $.
By the definition of~$ p $ we have $ (2/\epsilon')^p \leq n^{1/p} $ and thus $ (2/\epsilon')^q \leq (2/\epsilon')^p \leq n^{1/p} \leq n^{1/q} $ and furthermore, since $ p \leq \log{n} $ and $ R \geq n $, $ p \leq \log{n} \leq (\log{R})^q \leq (\log{R})^p = (2^p)^{\log{\log{R}}} \leq (n^{1/p})^{\log{\log{R}}} = n^{(\log{\log{R}})/p} \leq n^{(\log{\log{R}})/q} $.
It follows that the total update time is
\begin{equation*}
T_{q-2} (m, n) = O (m^{1 + O ((\log{\log{R}})/q)} R^{2/q})+ n) \, . \qedhere
\end{equation*}
\end{proof}

We can turn the algorithm above into an algorithm for the full distance range by using the rounding technique once more.

\begin{theorem}\label{thm:approximate_SSSP}
For every $ 0 < \epsilon \leq 1 $, there is a decremental approximate SSSP algorithm that, given a fixed source node $ s $, maintains, for every node $ v $, a distance estimate $ \distest (s, v) $ such that $ \dist_G (s, v) \leq \distest (s, v) \leq (1 + \epsilon) \dist_G (s, v) $.
It has constant query time and a total update time of
\begin{equation*}
O (m^{1 + O(\log^{5/4} ((\log{n}) / \epsilon) / \log^{1/4}{n}) } \log{W} + n)
\end{equation*}
in expectation.
If $ 1 / \epsilon = \polylog{n} $, then the total update time is $ O (m^{1 + o(1)} \log{W} + n) $ in expectation.
\end{theorem}

\begin{proof}
For every $ 0 \leq i \leq \lfloor \log(nW) \rfloor $ we define
\begin{equation*}
\varphi_i = \frac{\epsilon 2^i}{n} \, .
\end{equation*}
Let $ G_i' $ be the graph that has the same nodes and edges as $ G $ and in which every edge weight is rounded to the next multiple of $ \varphi_i $, i.e., every edge $ (u, v) $ in $ G_i' $ has weight
\begin{equation*}
w_{G_i'} (u, v) = \left\lceil \frac{w_G (u, v)}{\varphi_i} \right\rceil \cdot \varphi_i
\end{equation*}
where $ w_G (u, v) $ is the weight of $ (u, v) $ in $ G $.
This rounding guarantees that
\begin{equation*}
w_G (u, v) \leq w_{G_i} (u, v) \leq w_G (u, v) + \varphi_i
\end{equation*}
for every edge $ (u, v) $ of $ G $.
Furthermore we define $ G_i'' $ to be the graph that has the same nodes and edges as $ G_i' $ and in which every edge weight is scaled down by a factor of $ 1 / \varphi_i $, i.e., every edge $ (u, v) $ in $ G_i'' $ has weight
\begin{equation*}
w_{G_i''} (u, v) = \frac{w_{G_i'} (u, v)}{\varphi_i} = \left\lceil \frac{w (u, v)}{\varphi_i} \right\rceil \, .
\end{equation*}

The algorithm is as follows:
For every $ 0 \leq i \leq \lfloor \log(nW) \rfloor $ we use the algorithm of \Cref{lem:restricted_approximate_shortest_paths} on the graph $ G_i'' $ with $ R = 4 n / \epsilon $ to maintain a distance estimate $ \distest_i (s, v) $ for every node $ v $ that satisfies
\begin{itemize}
\item $ \distest_i (s, v) \geq \dist_{G_i''} (s, v) $ and
\item if $ \dist_{G_i''} (s, v) \leq R $, then $ \distest_i (s, v) \leq (1 + \epsilon) \dist_{G_i''} (s, v) $.
\end{itemize}
We let our algorithm return the distance estimate
\begin{equation*}
\distest (s, v) = \min_{0 \leq i \leq \lfloor \log{nW} \rfloor} \varphi_i \distest_i (s, v) \, .
\end{equation*}

We now show that there is some $ 0 \leq i \leq \lfloor \log(nW) \rfloor $ such that $ \varphi_i \distest_i (s, v) \leq (1 + 3 \epsilon) \dist_G (s, v) $.
As $ \distest (s, v) $ is the minimum of all the distance estimates, this implies that $ \distest (s, v) \leq (1 + 3 \epsilon) \dist_G (s, v) $.
In particular, we know that there is some $ 0 \leq i \leq \lfloor \log(nW) \rfloor $ such that $ 2^i \leq \dist_G (s, v) \leq 2^{i+1} $ since $ W $ is the maximum edge weight and all paths consist of at most $ n $ edges.
Consider a shortest path $ \pi $ from $ s $ to $ v $ in $ G $ whose weight is equal to $ \dist_G (s, v) $.
Let $ w_G (\pi) $ and $ w_{G_i'} (\pi) $ denote the weight of the path $ \pi $ in $ G $ and~$ G_i' $, respectively.
Since $ \pi $ consists of at most $ n $ edges we have $ w_{G_i'} (\pi) \leq w (\pi) + n \varphi_i $.
Therefore we get
\begin{align*}
\dist_{G_i'} (s, v) \leq w_{G_i'} (\pi) \leq w (\pi) + n \varphi_i = \dist_G (s, v) + \epsilon 2^i &\leq \dist_G (s, v) + \epsilon \dist_G (s, v) \\
 &= (1 + \epsilon) \dist_G (s, v) \, .
\end{align*}

Now observe the following:
\begin{align*}
\dist_{G_i''} (s, v) = \frac{\dist_{G_i'} (s, v)}{\varphi_i} \leq \frac{(1 + \epsilon) \dist_G (s, v)}{\varphi_i} \leq \frac{2 \dist_G (s, v)}{\varphi_i} &= \frac{2 \dist_G (s, v) n}{\epsilon 2^i} \\
 &\leq \frac{2 \cdot 2^{i+1} n}{\epsilon 2^i} = \frac{4 n}{\epsilon} = R \, .
\end{align*}
Since $ \dist_{G_i''} (s, v) \leq R $ we get $ \distest_i (s, v) \leq (1 + \epsilon) \dist_{G_i''} (s, v) $ by \Cref{lem:restricted_approximate_shortest_paths}.
Thus, we get
\begin{align*}
\varphi_i \distest_i (s, v) \leq \varphi_i((1 + \epsilon) \dist_{G_i''} (s, v)) = (1 + \epsilon) \dist_{G_i'} (s, v) &\leq (1 + \epsilon)^2 \dist_G (s, v) \\
 &\leq (1 + 3 \epsilon) \dist_G (s, v)
\end{align*}
as desired.

We now analyze the running time of this algorithm.
By \Cref{lem:restricted_approximate_shortest_paths}, for every $ 0 \leq i \leq \lfloor \log{(nW)} \rfloor $, maintaining $ \distest_i (s, v) $ on $ G_i'' $ for every node $ v $ takes time $ O (m^{1 + O ((\log{\log{R}})/q)} R^{2/q} + n) $, where
\begin{equation*}
q = \left\lfloor \sqrt{ \left\lfloor \frac{\sqrt{\log{n}}}{\sqrt{\log{\left( \frac{12 \cdot 4^3 \log{n}}{\epsilon} \right)}}} \right\rfloor } \right\rfloor
\end{equation*}
By our choice of $ R = 4 n / \epsilon $, the total update time for maintaining all these $ \lfloor \log{(nW)} \rfloor $ distance estimates is $ O (m^{1 + O ((\log{\log{(n / \epsilon)}})/q)} \log{W} / \epsilon) $ in expectation.
To obtain a $ (1 + \epsilon) $-approximation (instead of a $ (1 + 3 \epsilon) $-approximation, we simply run the whole algorithm with $ \epsilon' = \epsilon/3 $.
This results in a total update time of $ O (m^{1 + O ((\log{\log{(n / \epsilon)}})/q)} \log{W} / \epsilon) $, where
\begin{equation*}
q = \left\lfloor \sqrt{ \left\lfloor \frac{\sqrt{\log{n}}}{\sqrt{\log{\left( \frac{36 \cdot 4^3 \log{n}}{\epsilon} \right)}}} \right\rfloor } \right\rfloor
\end{equation*}
Now observe that $ 1 / \epsilon \leq n^{1/q} $ and that
\begin{align*}
O \left( \frac{\log{\log{\left( \frac{n}{\epsilon} \right)}}}{q} \right)
 = O \left( \frac{ \left( \log{\log{\left( \frac{n}{\epsilon} \right)}} \right) \left( \log \left( \frac{\log{n}}{\epsilon} \right) \right)^{1/4} }{(\log{n})^{1/4}} \right)
 = O \left( \frac{ \left( \log \left( \frac{\log{n}}{\epsilon} \right) \right)^{5/4} }{(\log{n})^{1/4}} \right) \, .
\end{align*}
The total update time therefore is
\begin{equation*}
O (m^{1 + O(\log^{5/4} ((\log{n}) / \epsilon) / \log^{1/4}{n}) } \log{W} + n) \, .
\end{equation*}
If $ 1 / \epsilon = \polylog{n} $, then the total update time is $ O (m^{1 + (\log^{5/4}{\log{n}}) / \log^{1/4}{n} } \log{W} + n) $, which is $ O (m^{1 + o(1)} \log{W} + n) $ since $ \lim_{x \to \infty} (\log^{5/4}{\log{n}}) / \log^{1/4}{n} = 0 $.

The query time of the algorithm described above is $ O (\log (nW)) $ as it has to compute $ \distest (s, v) = \min_{0 \leq i \leq \lfloor \log{nW} \rfloor} \varphi_i \distest_i (s, v) $ when asked for the approximate distance from $ s $ to $ v $.
We can reduce the query time to $ O (1) $ by using a min-heap for every node $ v $ that stores $ \distest_i (s, v) $ for all $ 0 \leq i \leq \lfloor \log(nW) \rfloor $.
This allows us to query for $ \distest (s, v) $ in constant time and does not increase our asymptotic bound on the total update time.
\end{proof}

\subsection{Approximate APSP}

We now show how to use our techniques to obtain a decremental approximate APSP algorithm.
This is conceptually simple now.
We use the approximate SSSP algorithm from \Cref{thm:approximate_SSSP} and plug it into the algorithm for maintaining approximate balls from \Cref{pro:from_sssp_to_balls}.
By using an adequate query procedure we can use the distance estimates maintained for the approximate balls to return the approximate distances between any two nodes.

\begin{theorem}
For every $ k \geq 2 $ and every $ 0 < \epsilon \leq 1 $, there is a decremental approximate APSP algorithm that upon a query for the approximate between any pair of nodes $ u $ and $ v $ returns a distance estimate $ \distest (u, v) $ such that $ \dist_G (u, v) \leq \distest (u, v) \leq {((2 + \epsilon)^k - 1)} \dist_G (u, v) $.
It has a query time of $ O (k^k) $ and a total update time of
\begin{equation*}
O (m^{1 + 1/k + O(\log^{5/4} ((\log{n}) / \epsilon) / \log^{1/4}{n}) } \log^2{W} + n)
\end{equation*}
in expectation.
If $ 1/\epsilon = \polylog{n} $, then the total update time is $ O(m^{1 + 1/k + o(1)} \log^2{W} + n) $ in expectation.
\end{theorem}

\begin{proof}
We use the approximate SSSP algorithm of \Cref{thm:approximate_SSSP} that provides a $ (1 + \epsilon) $-approximation and has an total update time of
\begin{equation*}
T (m, n) = O (m^{1 + O(\log^{5/4} ((\log{n}) / \epsilon) / \log^{1/4}{n}) } \log{W} + n)
\end{equation*}
in expectation, and if $ 1 / \epsilon = \polylog{n} $, then the total update time is $ T (m, n) = O (m^{1 + o(1)} \log{W} + n)$ in expectation.
By \Cref{pro:from_sssp_to_balls} we can maintain approximate balls with $ D = 5^k n W $ in total update time
\begin{equation*}
t (m, n, \p, \epsilon) = O \left( \left(\p m^{1 + 1/\p} + \sum_{0 \leq i \leq \p - 1} \frac{m}{m^{i/\p}} \cdot T (m_i, n_i) \right) \p \log{n} \frac{\log{(nW)}}{\epsilon} + \p \cdot T (m, n) \right)
\end{equation*}
in expectation, where, for each $ 0 \leq i \leq \p - 1 $, $ m_i = O (m^{(i+1)/\p}) $ and $ n_i = O (m^{(i+1)/\p}) $.
After simplification, we have $ t (m, n, \p, \epsilon) = O (m^{1 + 1/k + O(\log^{5/4} ((\log{n}) / \epsilon) / \log^{1/4}{n}) } \log^2{W} + n) $ and $ t (m, n, \p, \epsilon) = O (m^{1 + 1/k + o(1)} \log^2{W} + n) $ if $ 1/\epsilon = \polylog{n} $ as desired.

Additionally we maintain, for every node $ v \in V $ and every $ 1 \leq i \leq \p - 1 $, the node $ c_i (v) $ which is a node with minimum $ \distest (u, v) $ among all nodes $ u $ of priority $ i $ such that $ v \in B (u) $. 
This can be done as follows.
For every node $ v $ we maintain a heap containing all nodes $ u $ of priority $ i $ such that $ v \in B(u) $ using the key $ \distest (u, v) $.
Every time $ v $ joins or leaves $ B (u) $ we insert or remove $ u $ from the heap of $ v $.
Every time $ \distest (u, v) $ changes, we update the key of $ u $ in the heap of $ v $.
After each insert, remove, or update in the heap of some node $ v $, we find the minimal element $ c_i (v) $ of the heap.
As each heap operation takes logarithmic time, the total update time of the algorithm of \Cref{pro:from_sssp_to_balls} only increases by a logarithmic factor, which does not alter the overall running time bound of our algorithm.

\begin{procedure}
\caption{Query($u$, $v$)}
\label{alg:query}

\eIf{$ v \in B (u) $}{
	$ \distest' (u, v) \gets \distest (u, v) $\;
}{
	Set $ i $ to the priority of $ u $\;
	\ForEach{$ j = i + 1 $ \KwTo $ k - 1 $}{
		\eIf{$ c_j (u) $ exists}{
			$ v'' \gets c_j (u) $\;
			$ \distest' (v'', v) \gets $ \Query{$v''$, $v$}\;
			$ \distest_j' (u, v) \gets \distest (v, v'') + \distest' (v'', v) $\;
		}{
			$ \distest_j' (u, v) \gets \infty $\;
		}
	}
	$ \distest' (u, v) \gets \min_{i+1 \leq j \leq k-1} \distest_j' (u, v) $\;
}
\KwRet{$ \distest' (u, v) $}\;
\end{procedure}

To answer a query for the approximate distance between a pair of nodes $ u $ and $ v $ we use Procedure~\ref{alg:query}.
This procedure first tests whether $ v \in B (u) $ and if yes returns $ \distest (u, v) $.
Otherwise it does the following for every $ j \geq i+1 $, where $ i $ is the priority of $ u $:
It first computes the node $ c_j (u) $, which among the nodes $ v' $ of priority $ j $ with $ u \in B (v') $ is the one with the minimum value of $ \distest (v', u) $.
Then it recursively queries for the approximate distance $ \distest' (c_j (u), v) $ from $ c_j (u) $ to $ v $ and sets the distance estimate via $ c_j (u) $ to $ \distest_j' (u, v) =  \distest (v, c_j (u)) + \distest' (c_j (u), v) $.
Finally, it returns the minimum of all distance estimates $ \distest_j' (u, v) $.

Note that in each instance there are $ O (k) $ recursive calls and with each recursive call the priority of $ u $ increases by at least one.
Thus the running time of the query procedure is $ O(k^k) $.

\begin{claim}
For every pair of nodes $ u $ and $ v $ the distance estimate $ \distest' (u, v) $ computed by Procedure~\ref{alg:query} satisfies $ \distest' (u, v) \leq (((1 + \epsilon)^2 + 1)^{k-i} - 1) \dist_G (u, v) $, where $ i $ is the priority of $ u $.
\end{claim}

\begin{proof}
The proof is by induction on the priority $ i $ of $ u $.
Let $ \distest' (u, v) $ denote the distance estimate returned by Procedure~\ref{alg:query}.
If $ i = k-1 $, then by \Cref{pro:from_sssp_to_balls} and our choice of $ D $ we know that $ v \in B (u) $ and thus $ \distest' (u, v) = \distest (u, v) \leq (1 + \epsilon) \dist_G (u, v) $.
If $ i < k-1 $ we distinguish between the two cases $ v \in B (u) $ and $ v \notin B (u) $.
If $ v \in B (u) $, then $ \distest' (u, v) = \distest (u, v) \leq (1 + \epsilon) \dist_G (u, v) $.
If $ v \notin B (u) $, then by \Cref{pro:from_sssp_to_balls} and our choice of $ D $ there is a node $ v' $ of priority $ j > i $ such that $ u \in B (v') $ and $ \dist_G (u, v') \leq (1 + \epsilon)^2 ((1 + \epsilon)^2 + 1)^{j-i-1} \dist_G (u, v) $.

We will now argue that $ \distest_j' (u, v) \leq  2 ((1 + \epsilon)^3 + 1)^{k-1-i} - 1 ) \dist_G (u, v) $, which implies the same upper bound for $ \distest' (u, v) $.
Set $ v'' \gets c_j (u) $.
Since both $ v'' $ and $ v' $ have priority $ j $ and $ u \in B (v') $ as well as $ v \in B (v'') $ we have $ \distest (u, v'') \leq \distest (u, v') $ by the definition of~$ v'' $.
Since $ \distest (u, v') \leq (1 + \epsilon) \dist_G (u, v') $, we have
\begin{align*}
\distest (u, v'') \leq (1 + \epsilon) \dist_G (u, v') &\leq (1 + \epsilon)^3 ((1 + \epsilon)^2 + 1)^{j-i-1} \dist_G (u, v) \\
 &\leq (1 + \epsilon)^3 ((1 + \epsilon)^3 + 1)^{j-i-1} \dist_G (u, v) \, .
\end{align*}
To simplify the presentation in the following we set $ a = (1 + \epsilon)^3 $ and thus have $ \distest (u, v'') \leq a (a + 1)^{j-i-1} \dist_G (u, v) $.
By the triangle inequality we have
\begin{align*}
\dist_G (v'', v) \leq \dist_G (v'', u) + \dist_G (u, v) &\leq \distest (v'', u) + \dist_G (u, v) \\
 &\leq (a (a + 1)^{j-i-1} + 1) \dist_G (u, v)
\end{align*}
and by the induction hypothesis we have
\begin{align*}
\distest' (v'', v) &\leq (2 (a + 1)^{k-1-j} - 1) \dist_G (v'', v) \\
 &\leq (2 (a + 1)^{k-1-j} - 1) (a (a + 1)^{j-i-1} + 1) \dist_G (u, v) \, .
\end{align*}
Since $ j \geq i + 1 $ we get
\begin{align*}
\distest_j' (u, v) &= \distest (u, v'') + \distest' (v'', v) \\
 &\leq \left( a (a + 1)^{j-i-1} + (2 (a + 1)^{k-1-j} - 1) (a (a + 1)^{j-i-1} + 1) \right) \dist_G (u, v) \\
 &= \left( 2 (a + 1)^{k-1-j} (a (a + 1)^{j-i-1} + 1) - 1 \right) \dist_G (u, v) \\
 &= \left( 2 a (a + 1)^{k-1-(i+1)} + 2 (a + 1)^{k-1-j}) - 1 \right) \dist_G (u, v) \\
 &\leq \left( 2 a (a + 1)^{k-1-(i+1)} + 2 (a + 1)^{k-1-(i+1)}) - 1 \right) \dist_G (u, v) \\
 &= \left( 2 (a + 1)^{k-1-(i+1)} (a + 1) - 1 \right) \dist_G (u, v) \\
 &= ( 2 (a + 1)^{k-1-i} - 1 ) \dist_G (u, v) \, . \qedhere
\end{align*}
\end{proof}

Note that $ 2 \leq ((1 + \epsilon)^3 + 1) $ and therefore we have $ \distest' (u, v) \leq (((1 + \epsilon)^3 + 1)^{k-i} - 1) \dist_G (u, v) $.
Furthermore, $ (1 + \epsilon)^3 \leq 1 + 7 \epsilon $ and in the worst case $ i = 0 $.
Thus, by running the whole algorithm with $ \epsilon' = \epsilon / 7 $, we can guarantee that $ \distest' (u, v) \leq ((2 + \epsilon)^k - 1) \dist_G (u, v) $.
\end{proof}

%% file: conclusion.tex
\section{Conclusion}\label{sec:conclusion}

In this paper, we showed that single-source shortest paths in undirected graphs can be maintained under edge deletions with near-linear total update time and constant query time.
The main approach is to maintain an $(n^{o(1)}, \epsilon)$-hop set of near-linear size in near-linear time.
We leave two major open problems.
The first problem is whether the same total update time can be achieved for directed graphs, substantially improving the current $ m n^{0.9 + o(1)} $ total update time by \cite{HenzingerKNSTOC14,HenzingerKNICALP15}.
This problem is very challenging because a suitable hop set for directed graphs is not known even in the static setting.
Moreover, improving the current $ \tilde O (m \sqrt{n}) $ total update time by \cite{ChechikHILP16} for the decremental reachability problem is already very interesting. 

The second major open problem is to derandomize our algorithm.
The main task here is to deterministically maintain the priority-induced balls of the nodes up to small bounded distance, which is the key to maintaining the hop set.
A related question is whether the algorithm of Roditty and Zwick~\cite{RodittyZ12} for decrementally maintaining the priority-induced clusters of Thorup and Zwick~\cite{ThorupZ05} up to small bounded distance (and the corresponding spanners and emulators) can be derandomized, which is possible in the static setting \cite{RodittyTZ05}.
We have previously demonstrated~\cite{HenzingerKNSICOMP16} how to derandomize the decremental $ (1 + \epsilon) $-approximate APSP algorithm of Roditty and Zwick~\cite{RodittyZ12}, but the technique does not carry over to maintaining the clusters.
Using a principally different approach, Bernstein and Chechik~\cite{BernsteinC16} have recently introduced a technique to deterministically maintain approximate SSSP under edge deletions yielding total update times of $ \tilde O (n^2 \log{W}) $ in weighted graphs~\cite{Bernstein17} and $ \tilde O (m n^{3/4}) $ in unweighted graphs~\cite{BernsteinC17}, respectively.
Can their technique be extended to obtain a deterministic algorithm that is as fast as our randomized one in the sparse regime?